\documentclass[acmtocl,acmnow]{acmtrans2m}

\usepackage{amssymb,amsmath}  %% TM: added amsmath
\usepackage{graphicx}

\newcounter{myexamplecounter}

\newcounter{mypropositioncounter}
\newtheorem{proposition}[mypropositioncounter]{Proposition}
\newtheorem{example}[myexamplecounter]{Example}
\newdef{definition}[mydefinitioncounter]{Definition}
\newcommand{\definedterm}[1]{} % {\textsl{({#1}).}}
\newcommand{\examplegoal}[1]{} % {\textsl{({#1}).}}

\newcommand{\abbreviationpolice}[1]{{#1}}

\newcommand{\resp}{\abbreviationpolice{resp.}}
\newcommand{\cf}{\abbreviationpolice{cf.}}
\newcommand{\eg}{\abbreviationpolice{e.g.,}}
\newcommand{\etc}{\abbreviationpolice{etc.}}
\newcommand{\ie}{\abbreviationpolice{i.e.,}}
\newcommand{\wrt}{\abbreviationpolice{w.r.t.}}

\newcommand{\st}{\abbreviationpolice{s.t.}}
\newcommand{\vs}{\abbreviationpolice{vs.}}

\newcommand{\inconsistent}{\textsl{inconsistent}}
\newcommand{\implied}{\textsl{implied}}

\newcommand{\determined}{\textsl{determined}}
\newcommand{\dependent}{\textsl{dependent}}

\newcommand{\deepirrelevant}{\textsl{d-irrelevant}}  % Toni
\newcommand{\shallowirrelevant}{\textsl{s-irrelevant}} % Toni

\newcommand{\deepfixable}{\textsl{d-fixable}}
\newcommand{\shallowfixable}{\textsl{s-fixable}}
\newcommand{\deepremovable}{\textsl{d-removable}}
\newcommand{\shallowremovable}{\textsl{s-removable}}
\newcommand{\deepsubstitutable}{\textsl{d-substitutable}}
\newcommand{\shallowsubstitutable}{\textsl{s-substitutable}}
\newcommand{\deepinterchangeable}{\textsl{d-interchangeable}}
\newcommand{\shallowinterchangeable}{\textsl{s-interchangeable}}

\newcommand{\betweenparenth}[1]{%
  \left(\begin{array}{l} #1 \end{array}\right)}

\newcommand{\win}{\textsf{WIN}}
\newcommand{\out}{\textsf{out}}
\newcommand{\sol}{\textsf{sol}}
\newcommand{\sce}{\textsf{sce}}
\newcommand{\domain}{{\mathbb D}}
\newcommand{\implication}{\rightarrow}
\newcommand{\equivalence}{\leftrightarrow}

\title{%
  Generalizing Consistency and other Constraint Properties to Quantified Constraints
}
            
\author{
  LUCAS BORDEAUX \\ Microsoft Research \\
  MARCO CADOLI \\ Universit\`a di Roma ``La Sapienza'' \\
  TONI MANCINI \\ Universit\`a di Roma ``La Sapienza''}
            
\begin{abstract}
  Quantified constraints and Quantified Boolean Formulae are typically
  much more difficult to reason with than classical constraints,
  because quantifier alternation makes the usual notion of
  \emph{solution} inappropriate. As a consequence, basic properties 
  of Constraint Satisfaction Problems (CSP),
  such as consistency or substitutability, are not completely
  understood in the quantified case. These properties are important
  because they are the basis of most of the reasoning methods used to
  solve classical (existentially quantified) constraints, and one would like to
  benefit from similar reasoning methods in the resolution of quantified constraints.

  In this paper, we show that most of the properties that are used by
  solvers for CSP can be generalized to quantified CSP. 
  This requires a re-thinking of a number of basic concepts; in particular, 
  we propose a notion of \emph{outcome} that generalizes the classical notion
  of solution and on which all definitions are based.
  We propose a systematic study of the relations which hold between these
  properties, as well as complexity results regarding the decision of
  these properties.  Finally, and since these problems are typically
  intractable, we generalize the approach used in CSP and propose weaker, 
  easier to check notions based on \emph{locality}, which allow to detect
  these properties incompletely but in polynomial time.
\end{abstract}
            
\category{F4.1}{Mathematical Logic and Formal Languages}{Logic and Constraint Programming}
            
\terms{Algorithms} 
            
            
\begin{document}
            
\begin{bottomstuff} 
  Authors'address:
  Lucas Bordeaux, Microsoft Research Ltd, Roger
  Needham Building, J J Thomson Avenue, Cambridge CB3 0FB, United
  Kingdom. \newline
  Marco Cadoli and Toni Mancini, Universit\`a di Roma \emph{La Sapienza},
  Dipartimento di Informatica e Sistemistica, Via Salaria 113, 00198 Rome,
  Italy. \newline
  A preliminary version of this paper appears in \emph{Proc. of the 20th 
  National Conf. on Artificial Intelligence}, published by the American
  Association of Artificial Intelligence 
  \cite{Bordeaux-Cadoli-Mancini:AAAI:2005}. The current paper is a revised and
  extended version that includes proofs of all results.
\end{bottomstuff}
            
\maketitle

%%%%%%%%%%%%%%%%%%%%%%%%%%%%%%%%%%%%%%%%%%%%%%%%%%%%%%%%%%%%%%%%%%%%%%%%%%%%%%

\section{Introduction}

  \subsection{Quantified Constraints}

  Quantified Constraint Satisfaction Problems (QCSP) have recently received 
  increasing attention from the Artificial Intelligence community
  \cite{%
    Bordeaux-Monfroy:CP:2002,%
    Borner-Bulatov-Jeavons-Krokhin:CSL:2003,%
    Chen:AAAI:2004,%
    Chen:ECAI:2004,%
    Mamoulis-Stergiou:REPORT:2004,%
    Gent-Nightingale-Rowley:ECAI:2004,%
    Gent-Nightingale-Stergiou:IJCAI:2005,%
    Verger-Bessiere:CP:2006,%
    Benedetti-Lallouet-Vautard:IJCAI:2007,% 
    Bordeaux-Zhang:SAC:2007}.
  A large number of solvers are now available for Quantified Boolean
  Formulae (QBF), which represent the particular case of QCSP where
  the domains are Boolean and the constraints are clauses, see \eg\
  \cite{%
    Buening-Karpinski-Flogel:IC:1995,%
    Cadoli-Giovanardi-Schaerf:AAAI:1999,%
    Cadoli-Schaerf-Giovanardi-Giovanardi:JAR:2002,%
    Rintanen:IJCAI:1999}
  for early papers on the subject, and
  \cite{%
    Benedetti:LPAR:2004,%
    Zhang:AAAI:2006,%
    Samulowitz-etal:CP:2006,%
    Samulowitz-Bacchus:SAT:2006}
  for descriptions of state-of-the-art techniques for QBF. 
  The reason behind this trend
  is that QCSP and QBF are natural generalizations of CSP and SAT
  that allow to model a wide range of problems not directly
  expressible in these formalisms, and with applications in Artificial 
  Intelligence and verification.

  \subsection{Reasoning with Quantified Constraints}

  Quantified constraints are typically much more difficult to reason
  with than classical constraints. To illustrate this difficulty, let us start
  by an example of property we would like to characterize formally, and let us
  suggest why a number of naive attempts to define this property are not
  suitable. Consider the formula:
  \[
  \phi: ~~~ \forall x \in [3,10].~ \exists y \in [1,15]. ~ x = y.
  \] 
  We would like to ``deduce'' in a sense
  that $y \in [1,10]$ or, in other words, that the values $[11, 15]$ are 
  \emph{inconsistent} for $y$. Such a property will in particular be useful to
  a search-based solver: if this inconsistency is revealed, then the solver can
  safely save some effort by skipping the branches corresponding to the values 
  $y \in [11, 15]$.

  A first attempt to define this notion of consistency would be to use an
  implication and to say, for instance, that value
  $a$ is consistent for $y$ iff $\phi \rightarrow (y = a)$. 
  But there is clearly a problem with this approach since the occurrence of $y$ 
  on the right-hand-side of the implication is unrelated to its occurrences
  in formula $\phi$, which fall under the scope of a quantifier.
  One may attempt to circumvent this problem by putting the implication 
  under the scope of the quantifiers, and to say, for instance, that 
  $a$ is consistent for $y$ iff 
  $\forall x \in [3,10].~ \exists y \in [1,15]. ~ (x = y) \rightarrow (y = a$). 
  But with this definition any value would in fact be consistent, even $y=17$. 
  This is because for every $x$, we have a value for $y$ that falsifies the 
  left-hand side of the implication, thereby making the implication true. 

  Another approach that looks tempting at first but is also incorrect is
  to say that $a$ is inconsistent for $y$ iff the formula obtained 
  by fixing the domain of $y$ to $\{a\}$ is false. With this definition
  we would deduce that all values $a \in [1, 15]$ are inconsistent \wrt\ 
  variable $y$, since the formula $\forall x \in [3,10].~ \exists y \in
  [a,a]. ~ x = y$ is false in each and every case.
  Other variants of these definitions can be considered, but one quickly gets
  convinced that there is simply no natural way to define consistency, or any
  other property like \emph{interchangeability}, using implications or instantiations.  
%  The same problem holds for all the other properties used for automated reasoning 
%  on Constraint Satisfaction Problems, for instance interchangeability. 
  To define these notions properly in the case of quantified constraints,
  we need a new framework, which is what this paper proposes.

  \subsection{Overview of our Contributions}

  This paper shows that the definitions of consistency, substitutability, 
  and a wider range of CSP properties can be generalized to quantified constraints.
  Note that all our definitions and results also hold for the
  particular case of Quantified Boolean Formulas.
  These definitions, presented in Section \ref{sec:definitions-of-properties},
  are based on a simple game-theoretic framework and in particular on the new
  notion of \emph{outcome} which we identify as a key to define and
  understand all QCSP properties.
  We then classify these properties in Section \ref{sec:classification} 
  by studying the relationships between them (\eg\ some can be shown to be stronger than 
  others). We investigate the simplifications allowed by these properties in
  Section \ref{sec:simplifications},  and we characterize the
  complexity of their associated decision problem in Section \ref{sec:complexity}.   
  Since, as these complexity results show, determining whether
  any property holds is typically intractable in general, we investigate
  the use of the same tool which is used in classical CSP, namely
  \emph{local reasoning}, and we propose in Section \ref{sec:locality} local versions of 
  these properties that can be decided in polynomial time.
  Concluding comments follow in Section \ref{sec:conclusion}.
  We start (Section \ref{sec:Quantifier-Constraint-Satisfaction-Problems})
  by introducing some material on QCSP.

\section{Quantified Constraint Satisfaction Problems}
\label{sec:Quantifier-Constraint-Satisfaction-Problems}

  In this section, we present all the definitions related to QCSP,
  as well as some ``game-theoretic'' material. 

  %%%%%%%%%%%%%%%%%%%%%%%%%%%%%%%%%%%%%%%%%%%%%%%%%%%%%%%%%%%%%%%%%%%%%%%%%%%%%

  \subsection{Definition of QCSP}
  \label{subsec:Definition-of-QCSP}

  Let $\domain$ be a finite set. 
  Given a finite set $V$ of variables, a $V$-tuple $t$ with components in $\domain$, 
  is a mapping that associates a value
  $t_x \in \domain$ to every $x \in V$; a \emph{$V$-relation} over $\domain$ is a set
  of $V$-tuples with components in $\domain$.

  \begin{definition} \definedterm{QCSP}
    A \emph{Quantified Constraint Satisfaction Problem} (QCSP) is a
    tuple $\phi = \langle X, Q, D, C \rangle$ where: $X = \{x_1,
    \dots, x_n\}$ is a linearly ordered, finite set of
    \emph{variables}; $Q$ associates to each variable $x_i \in X$ a
    \emph{quantifier} $Q_{x_i} \in \{\forall, \exists\}$; $D$
    associates to every variable $x_i \in X$ a \emph{domain} $D_{x_i}
    \subseteq \domain$; and $C = \{c_1,\ldots c_m\}$ is a finite set of \emph{constraints},
    each of which is a $V$-relation with components in $\domain$ for some $V \subseteq X$.
    \label{definition-QCSP}
  \end{definition}
  
  \subsubsection{Notation}
  
  \begin{itemize}
  
  \item
  The notation $\prod_{x \in V} D_x$, where $V \subseteq X$ is a subset of variables, 
  will denote a \emph{Cartesian product} of domains, \ie\ the set of $V$-tuples
  $t$ that are such that $t_x \in D_x$ for each $x \in V$.

  \item
  The notation $t[x:=a]$, where $t$ is an $X$-tuple, $x \in X$ is a variable and 
  $a \in \domain$ is a value, will be used for \emph{instantiation}, \ie\ it denotes the
  tuple $t'$ defined by $t'_x = a$ and $t'_y = t_y$ for each $y \in X \setminus \{x\}$.
  
  \item
  The notation $t|_U$, where $t$ is a $V$-tuple and $U \subseteq V$ is a subset of its 
  variables, will denote the \emph{restriction} of $t$ to $U$, \ie\ the $U$-tuple $t'$
  such that $t'_x = t_x$ for each $x \in U$. (Note that $t$ is undefined on every 
  $y \in V \setminus U$.)
  
  \end{itemize}
  
  We use the following shorthands to denote the set of existential (\resp\
  universal) variables, the set of variables of index $\leq j$,
  and the sets of existential/universal variables of index $\leq j$:
  \[
  \begin{array}{rclrcl}
    & & &  ~~
    X_j & \!\!\!=\!\!\! & \{x_i \in X ~|~ i \leq j\}
    \\
    E   & \!\!\!=\!\!\! & \{x_i \in X ~|~ Q_{x_i} \!=\! \exists\} & 
    E_j & \!\!\!=\!\!\! & E \cap X_j
    \\
    A   & \!\!\!=\!\!\! & \{x_i \in X ~|~ Q_{x_i} \!=\! \forall\} & 
    A_j & \!\!\!=\!\!\! & A \cap X_j
  \end{array}
  \]
  
  \subsubsection{Satisfaction, Solutions and Truth of a QCSP}
  \label{subsec-truth}
  Given a QCSP $\phi = \langle X, Q, D, C \rangle$ as in Definition~\ref{definition-QCSP},
  an $X$-tuple $t$ is
  said to \emph{satisfy} the set of constraints $C$ if $t|_V \in c$
  for each $V$-relation $c \in C$. The set of $X$-tuples satisfying
  all constraints of $\phi$ is called the set of \emph{solutions} to $C$ and is
  denoted by $\sol^\phi$. 
  
  Although QCSPs are defined in a form that closely follows the traditional
  definition of CSPs, the most immediate way to define their semantics is to use
  rudimentary logic with equality. (We shall see in the next section that
  we can in a second step forget about the logic and think alternatively 
  in terms of tuples and functions when this is more convenient.)
  A QCSP $\langle X, Q, D, C \rangle$ represents a logical formula whose vocabulary
  includes $n$ names for the variables (for convenience, we simply denote these names
  as $x_1 \dots x_n$) and $m$ names for the constraints ($c_1 \dots c_m$).
  The formula is defined as:
  \[
  F: ~~
  Q_{x_1} x_1 \in D_{x_1} \dots Q_{x_n} x_n \in D_{x_n} ~ 
  (F_1 \wedge \dots \wedge F_m).
  \]
  where each $F_i$ is obtained from the corresponding $V$-relation $c_i$: let
  $\{y_1, .., y_p\} = V$, then $F_i$ is simply the formula $c_i(y_1, .., y_p)$, 
  \ie\ we apply the name of the constraint to the right argument list.
  Each $D_{x_i}$ explicitly lists the values specified in the QCSP definition,
  for instance $\forall x \in \{a,b\}. \phi$ is a shorthand for 
  $\forall x. (x=a \vee x=b) \rightarrow \phi$.
  
  Let $I$ be the interpretation function that associates to each constraint
  name the corresponding relation; the QCSP is said to be \emph{true} if
  formula $F$ is true in the domain $\domain$ and \wrt\ the interpretation $I$,
  \ie\ iff 
  $
  \langle \domain, I\rangle \models F
  $.

  %%%%%%%%%%%%%%%%%%%%%%%%%%%%%%%%%%%%%%%%%%%%%%%%%%%%%%%%%%%%%%%%%%%%%%%%%%%%%

  \subsection{Game-Theoretic Material}
  \label{subsec:Game-theoretic-material}

  Quantifier alternation is best understood using an ``adversarial'' or
  ``game-theoretic'' viewpoint, where two players interact. One of them is allowed to
  choose the values for the existential variables, and its aim is to
  ultimately make the formula true, while the other assigns the
  universal variables and aims at falsifying it. 
  We introduce several definitions leading to our central notion of
  \emph{outcome}, which will be shown to shed light on the definition
  of properties in the next section. Our presentation of the basic 
  game-theoretic material is 
  inspired from \cite{Chen:ECAI:2004}, who uses a similar notion of winning
  strategy.
  
  The following QCSP (written using the usual, self-explanatory logical notation 
  rather than in the form of a tuple $\langle X, Q, D, C \rangle$) 
  will be used to illustrate the notions throughout this sub-section:
  
  \begin{equation}
  \begin{array}{l}
    \begin{array}{r}
      \exists x_1 \in [1,10]. ~
      \forall x_2 \in [1,10]. ~\exists x_3 \in [1,10]. \\
      \forall x_4 \in [1,10].  ~\exists x_5 \in [1,10].
    \end{array} ~
    \begin{array}{r} ~ \\
      x_1 + x_2 + x_3 + x_4 + x_5 = 30
    \end{array}
  \end{array}
  \label{eq-main-example}
  \end{equation}
  
  This formula can be thought of as a game between two players 
  assigning, respectively, the odd and even variables. The players
  draw in turn between 1 and 10 sticks from a heap containing
  originally 30 sticks; the player who takes the last stick
  wins. 
  
  \subsubsection{Strategies}
  
  The first notion we need is the notion of \emph{strategy}:

  \begin{definition} \definedterm{strategy}
    A strategy is a family $\{s_{x_i} ~|~ x_i \in E\}$ where each $s_{x_i}$
    is a function of signature 
    $\left(\prod_{y \in A_{i-1}} D_y\right) \rightarrow D_{x_i}$.
  \end{definition}

  In other words, a strategy defines for each existential variable $x_i$ a function
  that specifies which value to pick for $x_i$ depending on the values 
  assigned to the universal variables that precede it.
  Note in particular that, if the first $k$ variables of the problem are
  quantified existentially, we have for every $i \leq k$ a constant
  $s_{x_i} \in D_{x_i}$ which defines which value should directly be
  assigned to variable $x_i$.
  \begin{example}%
    A strategy for the QCSP (\ref{eq-main-example})
    can be defined by $s_{x_1}() = 8$; 
    $s_{x_3}$ associates to every $\{x_2\}$-tuple $t$ the value
    $s_{x_3}(t) = 11 - t_{x_2}$ 
    and $s_{x_5}$ associates to every $\{x_2, x_4\}$-tuple $t$ the value
    $s_{x_5}(t) = 11 - t_{x_4}$. 
    This strategy specifies that we first draw 8 sticks, 
    then for the next moves we shall draw 11 minus what the opponent just drew. 
    \label{ex-strategy-example}
  \end{example}

  \subsubsection{Scenarios}

  The tuple of values that will eventually be
  assigned to the variables of the problem depends on two things: 1)
  the strategy we have fixed \emph{a priori}, and 
  2) the sequence of choices of the ``adversary'',
  \ie\ the values that are assigned to the universal variables. Given a 
  particular strategy, a number of potential \emph{scenarios} may therefore
  arise, depending on what the adversary will do. These scenarios are
  defined as follows:

  \begin{definition} \definedterm{scenario}
    The set of scenarios of a strategy $s$ for a QCSP $\phi$, denoted
    $\sce^\phi(s)$, is the set of tuples $t \in \prod_{x \in X} D_x$ 
    such that, for each $x_i \in E$, we have:
    \[
      t_{x_i} = s_{x_i}(t|_{A_{i-1}})
    \]
  \end{definition}

  In other words, the values for the existential variables are
  determined by the strategy in function of the values assigned to the
  universal variables preceding it. There is no restriction, on the
  contrary, on the values assigned to universal variables: this reflects the fact
  that we
  model the viewpoint of the existential player, and the adversary may play
  whatever she wishes to play.
  
  \setcounter{myexamplecounter}{0}
  \begin{example}(Ctd.)
    An example of scenario for the strategy defined previously
    is the tuple defined by $x_1 = 8, x_2 = 4, x_3 = 7, x_4 = 1, x_5 =
    10$. On the contrary, the tuple $x_1 = 8, x_2 = 4, x_3 = 7, x_4 =
    1, x_5 = 5$ is not a scenario since the value 5 for $x_5$ does not
    respect what is specified by $s_{x_5}$.
  \end{example}

  \subsubsection{Winning Strategies}
  
  Of particular interest are the strategies whose scenarios are all solutions. We
  call them \emph{winning strategies}:

  \begin{definition} \definedterm{winning strategy}
    A strategy $s$ is a winning strategy for the QCSP $\phi$ if every
    scenario $t \in \sce^\phi(s)$ satisfies the constraints of $\phi$
    (in other words: if $\sce^\phi(s) \subseteq \sol^\phi$).
  \end{definition}

  We denote by $\win^\phi$ the set of winning strategies of the QCSP
  $\phi$.
  
  \setcounter{myexamplecounter}{0}
  \begin{example}(Ctd.)
    In the strategy $s$ defined in Example \ref{ex-strategy-example}, 
    any scenario $t$ is of the form
    $x_1=8, x_2=a, x_3=11-a, x_4 = b, x_5=11-b$. As a result the sum 
    always evaluates to $8 + a + 11-a + b + 11-b  = 30$ 
    and $s$ is therefore a winning strategy.
    In fact, this strategy is the only winning one; one can check, for instance,
    that the strategy $s'$ defined by $s'_{x_1}() = 7$; $s'_{x_3}(t) = 7$ 
    and $s'_{x_5}(t) = 7$ is not winning.
  \end{example}
  
  The following proposition is essential in that it justifies the use of
  the game-theoretic approach\footnote{%
    Proofs of all propositions can be found in the online Appendix \ref{app-proofs}.
    }:
  
  \begin{proposition}
    A QCSP is true  (as defined in Section \ref{subsec-truth}) 
    iff it has a winning strategy.
    \label{prop-qcsp-true-iff-win}
  \end{proposition}

  \subsubsection{Outcome}
  
  Whereas the preceding material is well-known and is used, for instance, in
  \cite{Chen:ECAI:2004}, we introduce the following new notion:

  \begin{definition} \definedterm{outcome}
    The set of outcomes of a QCSP $\phi$, denoted $\out^\phi$,
    is the set of all scenarios
    of all its winning strategies, \ie\ it is defined as:
    \[
    \out^\phi
    ~=~ 
    \bigcup_{s \in \win^\phi}   \sce^\phi (s)
    \]
  \end{definition}
  
  \setcounter{myexamplecounter}{0}
  \begin{example}(Ctd.)
    Since our example has a unique winning strategy it is easy to
    characterise its set of outcomes: these are all the tuples of the
    form $x_1=8, x_2=a, x_3=11-a, x_4 = b, x_5=11-b$, 
    with $a,b \in [1,10]$.
  \end{example}

  Outcomes are related to the classical notion of solution in the following
  way: in general any outcome satisfies the set of constraints $C$, so we have
  $\out^\phi \subseteq \sol^\phi$, and the equality  $\out^\phi = \sol^\phi$ holds
  if all variables are existential. On the other hand let us emphasize the fact
  that not all solutions are necessarily outcomes in general: in our example the
  tuple $x_1=6, x_2=6, x_3=6, x_4 = 6, x_5=6$ is for instance a solution as
  it satisfies the unique constraint ($x_1 + x_2 + x_3 + x_4 + x_5 = 30$).
  But there is no winning strategy whose set of scenarios includes this particular
  tuple, and it is therefore not an outcome.
  
  The notion of outcome is a generalization of the notion of
  solution that takes into account the quantifier prefix of the constraints.
  Our claim in the following is that \emph{outcomes play a role as central for
  QCSP as the notion of solution does in CSP, and that most definitions can be
  based on this notion}.
  
  \subsubsection{Summary of the notions and notations}

  To summarize, we have defined 3 sets of tuples ($\sol^\phi$: the set
  of solutions, $\sce^\phi(s)$: the set of scenarios of strategy $s$, 
  and $\out^\phi$: the set of outcomes) and one set of strategies 
  ($\win^\phi$: the set of winning strategies).
  All the game-theoretic notions we have introduced
  are illustrated in Fig. \ref{fig:killer-example}, 
  where %
  %
  %
  % In Fig. \ref{fig:killer-example} 
  we consider the QCSP represented by the 
  logical formula:
  \begin{equation}
  \exists x_1 \in [2,3] ~ \forall x_2 \in [3,4] ~\exists
  x_3 \in [3,6]. ~x_1 + x_2 \leq x_3.
  \label{eq-qcsp-example}
  \end{equation}
  \emph{And} and \emph{or}
  labels on the nodes correspond to universal and existential
  quantifiers, respectively. The solutions are all triples $\langle
  x_1, x_2, x_3 \rangle$ \st\ $x_1 + x_2 \leq x_3$. The only two
  winning strategies assign $x_1$ to $2$: one ($s_1$) systematically assigns $x_3$
  to 6 while the 2nd one ($s_2$) assigns it to $x_2 + 2$ (note that
  each strategy is constrained to choose one unique branch for each
  existential node). The scenarios of $s_1$ and $s_2$ are therefore
  those indicated, while the set of outcomes of the QCSP is the
  union of the scenarios of $s_1$ and $s_2$ (also shown in bold
  line).

 \begin{figure}[h]
  \centering
  \includegraphics[scale=.65, bb=120 282 455 568]{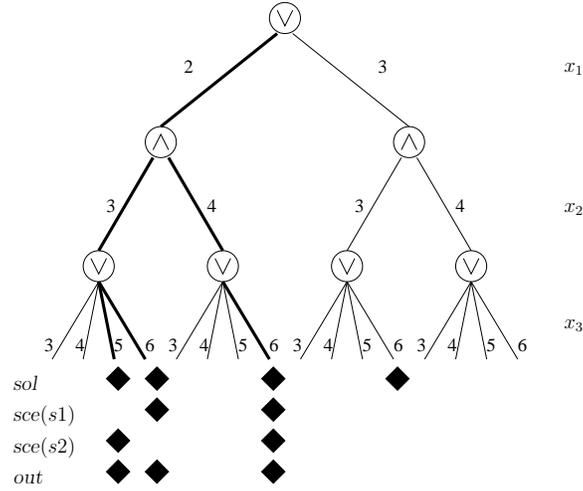}
  \caption{%
      A summary of the game-theoretic notions used in this paper.
   }
  \label{fig:killer-example}
  \end{figure}

\section{Definitions of the CSP Properties}
\label{sec:definitions-of-properties}

  \subsection{Informal Definitions of the Properties}

  A major part of the CSP literature aims at identifying properties of
  particular values of some variables. The goal is typically to
  simplify the problem by ruling out the possibility that a variable
  $x_i$ can be assigned to a value $a$. This can be done when
  one of the following properties holds, with respect to variable $x_i$:
  
  \begin{itemize}
  
  \item Value $a$ is guaranteed not to participate in any solution: $a$ is
  \emph{inconsistent} for $x_i$ \cite{Mackworth:AI:1977}.
  
  \item
  Another value $b$ can replace $a$ in any solution involving it: $a$ is
  \emph{substitutable} to $b$ for $x_i$ \cite{Freuder:AAAI:1991}.
  
  \item All solutions involving $a$ can use another value instead: $a$ is
  \emph{removable} for $x_i$ \cite{Bordeaux-Cadoli-Mancini:LPAR:2004}.
  
  \end{itemize}

  On the contrary, some other properties give an indication that
  instantiating $x_i$ to $a$ is a good idea:
  
  \begin{itemize}
  
  \item All solutions assign value $a$ to variable $x_i$:
  $a$ is  \emph{implied} for $x_i$ 
  \cite{Monasson-Zecchina-Kirkpatrick-Selman-Troyansky:NATURE:1999};
  
  \item We have the guarantee to find a solution with value $a$
  on $x_i$, if a solution exists at all: $a$ is said to be
  \emph{fixable} for $x_i$ \cite{Bordeaux-Cadoli-Mancini:LPAR:2004}. 
  
  \end{itemize}
  
  While all the preceding are properties of particular \emph{values}, 
  related properties of \emph{variables} are also of interest: 
  
  \begin{itemize}
  
  \item
  The value assigned to a variable $x_i$ is forced to a unique
  possibility: $x_i$ is \emph{determined}.
  
  \item
  The value of variable $x_i$ is a function of
  the values of other variables: $x_i$ is \emph{dependent}.
  
  \item 
  Whether a tuple is a solution or not does not depend on the value assigned
  to  variable $x_i$: $x_i$ is \emph{irrelevant}.
  
  \end{itemize}

  In this section, we propose generalizations of the definitions of
  the main CSP properties to quantified constraints. 
  For the sake of homogeneity, we adopt the terminology used in the paper
  \cite{Bordeaux-Cadoli-Mancini:LPAR:2004} for the names of the properties.
  
  We adopt a predicate notation and write, \eg\ $\textsl{p}^\phi(x_i, a)$ for the
  statement ``value $a$ has property \textsl{p} for variable $x_i$ (in
  QCSP $\phi$)''. 
  The superscript $\phi$ will be omitted in order to simplify the notation
  whenever there is no ambiguity regarding which QCSP is considered.
  
  We present our definitions in two steps: 
  Section~\ref{subsec-basic-definitions} introduces
  the basic definitions, which we call \emph{deep} definitions, for reasons that
  will become clear in the rest of this section. We then notice in Section \ref{subsec:shallow}
  that the properties can be made more general, leading to our \emph{shallow} definitions.

  \subsection{Basic Definitions}
  \label{subsec-basic-definitions}

  The first definitions we propose are identified by a \textsl{d} prefix
  and qualified as ``deep'' when an ambiguity with the definitions in forthcoming 
  Section~\ref{subsec:shallow} is possible.
  They are based on directly rephrasing the original CSP definitions, but using
  the notion of outcomes in place of solutions:

  \begin{definition} \definedterm{``Deep'' properties}
    We define the properties of inconsistency, implication, deep fixability,
    deep substitutability, deep removability, deep interchangeability, 
    determinacy, deep irrelevance and dependency, as follows, 
    for all $x_i \in X$, $a, b \in D_{x_i}$, $V \subseteq X$:
    \[
    \begin{array}{rll}
      \inconsistent(x_i, a) & \equiv~~ &
        \forall t \in \out. ~~ t_{x_i} \not= a
      \\
      \implied(x_i, a) & \equiv &
        \forall t \in \out. ~~ t_{x_i} = a
      \\
      \\
      \deepfixable(x_i, a) & \equiv &
        \forall t \in \out. ~~ t[x_i := a] \in \out
      \\
      \\ 
      \deepsubstitutable(x_i, a, b) & \equiv &
        \forall t \in \out. ~~ (t_{x_i} = a) \implication (t[x_i:=b] \in \out)
      \\
      \\
      \deepremovable(x_i, a) & \equiv &
        \forall t \in \out. ~~
          (t_{x_i} = a)  \implication (\exists b \neq a.
          ~~ t[x_i := b] \in \out)
      \\
      \\
      \deepinterchangeable(x_i, a, b) & \equiv & 
         \deepsubstitutable(x_i, a, b) \wedge \deepsubstitutable(x_i, b, a)
      \\
      \\
      \determined(x_i) & \equiv &
         \forall t \in \out. ~~\forall b \not= t_{x_i}. ~~
           t[x_i := b] \not\in \out
      \\
      \deepirrelevant(x_i) & \equiv &
        \forall t \in \out. ~~\forall b \in D_{x_i}. ~ t[x_i := b] \in \out
      \\
      \\
      \dependent(V, x_i) & \equiv &
         \forall t, t' \in \out. ~~
           (t|_V = t'|_V) \implication (t_{x_i} = t'_{x_i})
%         \big( \forall x_j \in V. ~ t_{x_j} = t'_{x_j}\big)
%         \implication (t_{x_i} = t'_{x_i})
    %
    \end{array}
    \]
    \label{def:deep-properties}
  \end{definition}

  We note that the definition of consistency is equivalent to the one proposed 
  in \cite{Bordeaux-Monfroy:CP:2002}; it is nevertheless expressed in a
  simpler and more elegant way that avoids explicitly dealing with
  And/Or trees. All other definitions are new. 
%  We now exemplify some of the definitions:
  
  \begin{example} \examplegoal{Illustration of Def. \ref{def:deep-properties}}
    Consider the QCSP:
    \[
    \exists x_1 \in [2,3] ~ \forall x_2 \in [3,4]
    ~\exists x_3 \in [3,6].  ~x_1 + x_2 \leq x_3
    \]
    (\cf\
    Fig. \ref{fig:killer-example}). We have: 
    $\inconsistent(x_1, 3)$,  
%    $\deepfixable(x_1, 2)$, $\implied(x_1, 2)$,
%    $\deepsubstitutable(x_1, 3, 2)$,
    $\inconsistent(x_3, 3)$,
    $\inconsistent(x_3, 4)$,
    $\deepsubstitutable(x_3, 5, 6)$, $\deepfixable(x_3,6)$,
    $\deepremovable(x_3,5)$, and $\implied(x_1, 2)$.
    \label{example-deep}
  \end{example}

  A choice we made in Definition \ref{def:deep-properties} requires a justification: 
  if we consider, for instance,
  fixability, one may think that a more general definition could be obtained if we
  wrote $\forall t \in \out. ~t[x_i := a] \in \underline{\sol}$ instead of 
  $\forall t \in \out. ~t[x_i := a] \in {\out}$. 
  Similarly, the question arises whether the other definitions that involve 
  the set \out\ in the right-hand side of an implication (either implicitly or 
  explicitly) could be strengthened be using the set \sol\ instead.
  This is not the case: except for one property, namely \emph{determinacy},
  the modified definitions would actually be strictly equivalent:
  
  \begin{proposition}
    Deep fixability could equivalently be defined by the condition
    $\forall t \in \out. t[x_i := a] \in \sol$;
    Deep substitutability could be equivalently defined by 
    $\forall t \in \out.$  $(t_{x_i} = a) \implication (t[x_i:=b] \in \sol)$;
    deep removability by 
    $\forall t \in \out.
          (t_{x_i} = a)  \implication (\exists b \neq a.
          t[x_i := b] \in \sol)$;
    and deep irrelevance by
    $\forall t \in \out. \forall b \in D_{x_i}. ~ t[x_i := b] \in \sol$.
    \label{prop-no-outcome}
  \end{proposition}
  
  This proposition will play a role in the proof of other results.
  Defining determinacy by $\forall t \in \out. \forall b \not= t_{x_i}.
  t[x_i := b] \not\in \sol$, instead of the definition we used. \ie\ 
  $\forall t \in \out. \forall b \not= t_{x_i}. t[x_i := b] \not\in \out$,
  would on the contrary give a slightly different notion: we note that
  in this case (because of the negation implicitly on the right-hand
  side of the implication, \ie\ $t[x_i := b] \not\in \out$), the definition
  would become \emph{weaker}. For instance, in Fig. \ref{fig:killer-example}, we
  would not have $\determined(x_1, 2)$ because the tuple
  $t = \langle 2, 3, 6\rangle$ is such that $t[x_1 := 3] \in \sol$.

  \subsection{Generalization: Shallow Definitions}
  \label{subsec:shallow}

  The previous definitions are correct in a sense that will be made
  formal in Section~\ref{sec:simplifications}. They are nevertheless
  overly restrictive in some cases, as the following example shows:

  \begin{example}
    Consider the QCSP:
    \[
    \forall x_1 \in [1,2] ~ \exists x_2 \in [3,4]
    ~\exists x_3 \in [4,6].  ~x_1 + x_2 = x_3.
    \] 
    The winning strategies can
    make arbitrary choices for $x_2$ as long as they set $x_3$ to the value
    $x_1 + x_2$, and the outcomes are the triples $\langle
    1,3,4\rangle$, $\langle 1,4,5\rangle$, $\langle
    2,3,5\rangle$, $\langle 2,4,6\rangle$. Note that for variable $x_2$,
    neither values 3 nor 4 are deep-fixable, and none is deep-substitutable
    to the other. This somehow goes against the intuition that we are
    indeed free to choose the value for $x_2$.
    \label{example-shallow}
  \end{example}

  The reason why our previous definition did not capture this case is
  that it takes into account the values of the variables occurring
  \emph{after} the considered variable: values 3 and 4 are
  interchangeable (for instance) only if the QCSPs resulting from
  these instantiations can be solved \emph{using the same strategy}
  for all the subsequent choices---this is why
  we called these definitions \emph{deep} (with a \textsl{d} prefix).
  On the contrary, we can formulate \emph{shallow} definitions of the
  properties, which accept value 4 as a valid substitute for 3 because
  \emph{in any sequence of choices leading to the possibility of
  choosing 3 for $x_2$, value 4 is also a valid option}.

  \begin{definition} \definedterm{``shallow'' properties}
    We define the properties of shallow fixability, 
    substitutability, removability, interchangeability, and
    irrelevance, as follows:
    \noindent
    \[
    \begin{array}{l}
      \shallowfixable(x_i, a) \equiv \\ \hspace*{.5cm}
        \forall t \in \out. ~
        \exists t' \in \out. ~
        \betweenparenth{
           t|_{X_{i-1}} = t'|_{X_{i-1}}  %\\
           \wedge ~~ t'_{x_i} = a
        }
    \end{array}
    \]
    
    \[
    \begin{array}{l}
      \shallowsubstitutable(x_i, a, b) \equiv \\ \hspace*{.5cm}
        \forall t \in \out. ~
        t_{x_i} = a \implication \\ \hspace*{1cm}
        \exists t' \in \out. ~
        \betweenparenth{
          (t|_{X_{i-1}} = t'|_{X_{i-1}})
          ~\wedge~ (t'_{x_i} = b)
        }
    \\ \\
      \shallowremovable(x_i, a) \equiv \\ \hspace*{.5cm}
        \forall t \in \out. ~
        t_{x_i} = a \implication \\ \hspace*{1cm}
        \exists t' \in \out. ~
        \betweenparenth{
          t|_{X_{i-1}} = t'|_{X_{i-1}} \land t'_{x_i} \neq a
        }
    \\ \\
      \shallowinterchangeable(x_i, a, b) \equiv \\ \hspace*{.5cm}
         \shallowsubstitutable(x_i, a, b) \wedge % \\ \hspace*{.5cm}
         \shallowsubstitutable(x_i, b, a)
    \\ \\
      \shallowirrelevant(x_i) \equiv \\ \hspace*{.5cm}
        \forall t \in \out. ~
        \forall b \in D_{x_i}.  \\ \hspace*{1cm}
          \exists t' \in \out. ~
            \betweenparenth{
              (t|_{X_{i-1}} = t'|_{X_{i-1}})
                 ~\wedge~ (t'_{x_i} = b)
            }
    \end{array}
    \]
    \label{def:shallowProperties}
  \end{definition}

  One can check that with these definitions we handle Example 
  \ref{example-shallow} as expected:
  
  \setcounter{myexamplecounter}{2}
  \begin{example}(Ctd.)
    Considering again the QCSP:
    \[
    \forall x_1 \in [1,2] ~ \exists x_2 \in [3,4]
    ~\exists x_3 \in [4,6].  ~x_1 + x_2 = x_3,
    \]
    values 3 and 4 are shallow-interchangeable for variable $x_2$
    (both values are also shallow-fixable, shallow-removable, and variable $x_2$ is in fact
    shallow-irrelevant). The reason is that for each outcome $t$
    that assigns value 3 to $x_2$, there exists a tuple $t'$ such that
    $t'_{x_1} = t_{x_1}$ and $t'_{x_2} = 4$ (to $t = \langle 1,3,4\rangle$
    corresponds  $t' =\langle 1,4,5\rangle$; to 
    $\langle 2,3,5\rangle$ corresponds $\langle 2,4,6\rangle$), and vice-versa.
    
    This can be seen pictorially in Fig. \ref{fig:shallow}. On the
    left-hand side, we see why values 3 and 4 are not (for instance) 
    deep-interchangeable for $x_2$: the outcomes (branches) going through these values are 
    indeed different. Now on the right-hand side we see the viewpoint of
    the \emph{shallow} definitions: the strategy is only considered \emph{up to
    variable $x_2$}, and it is clear, then, that values 3 and 4
    are interchangeable.
  \end{example}

  \begin{figure}[htbp]
  \centering
  \includegraphics[scale=.5]{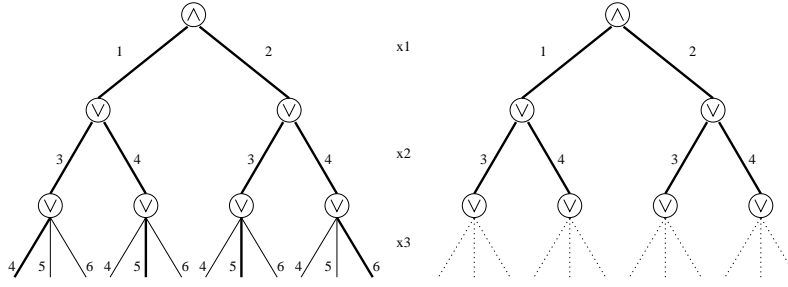}
  \caption{%
      Illustration of the notion of \emph{shallow} properties, as opposed to the \emph{deep}
      definitions.
   }
  \label{fig:shallow}
  \end{figure}
  
  We last remark that the distinction we have introduced between \emph{deep} and
  \emph{shallow} only makes sense for a subset of the properties. It is easy to see,
  for instance, that a shallow definition of \emph{inconsistency} would make no difference:
  this notion is defined by the statement $\forall t \in \out. ~~ t_{x_i} \not= a$, and
  this is equivalent to $\forall t \in \out. ~~ (t|_{X_{i}})_{x_i} \not= a$.

\section{Relations between the Properties}
%\section{Correctness of the definitions and relations between them}
\label{sec:classification}

  This section gives a number of results establishing the relations between
  the classes of properties (\eg\ deep, shallow) and between the properties
  themselves (substitutability, determinacy, \etc). 
  These results will also be used later (Section \ref{sec:simplifications})
  to prove that our definitions are useful, in that they allow 
  to simplify the considered QCSP while preserving some form of equivalence.
  
  \subsection{Relations between Classes of Properties}
  
  The basic relations between classical, deep, and shallow definitions, are the
  following: deep definitions are \emph{more general} than basic,
  existential ones, and the shallow definitions are \emph{more general} than the
  deep ones, in a sense that is explained formally in the following.
  
  \subsubsection{Deep definitions \vs\ classical definitions}

  We first note that, in the particular case where the quantifiers are all
  existential, the deep definitions of the properties (Definition~\ref{def:deep-properties}) 
  correspond to the classical CSP notions, simply because we have \out\ = \sol\ in that 
  case; in other words our definitions truly are \emph{generalizations} of the classical 
  definitions.
  In the general case, when the quantifiers are not restricted to be existential,
  we can still ignore the quantifier prefix and apply the classical definitions to the
  resulting existentially quantified CSP. The relations
  between the original QCSP and the relaxed CSP are the following:
  
  \begin{enumerate}

  \item
  The deductions made using the classical definitions are \emph{correct}:
  a property detected on the existentially quantified CSP, using the classical 
  definitions, will also hold for the QCSP.

  \item 
  This reasoning is \emph{incomplete}:
  if we do not take into account the quantifier prefix as our new definitions do,
  some properties cannot be detected.

  \end{enumerate}
  
  The \emph{correctness} can be stated formally as follows:
  
  \begin{proposition}
    Let $\phi = \langle X, Q, D, C\rangle$ be a QCSP and let 
    $\psi$ be the same QCSP but in which all quantifiers are existential,
    \ie\ $\psi = \langle X, Q', D, C\rangle$, with $Q'_{x} = \exists$, for all
    $x \in X$. We have (forall $x_i, a, b, V$):
    
    \begin{itemize}
    
    \item $\inconsistent^\psi(x_i, a) \implication \inconsistent^\phi(x_i, a)$;
    
    \item $\deepfixable^\psi(x_i, a) \implication \deepfixable^\phi(x_i, a)$;
    
    \item $\deepsubstitutable^\psi(x_i, a, b) \implication \deepsubstitutable^\phi(x_i, a, b)$;
    
    \item $\deepremovable^\psi(x_i, a) \implication \deepremovable^\phi(x_i, a)$;
    
    \item $\deepinterchangeable^\psi(x_i, a, b) \implication \deepinterchangeable^\phi(x_i, a, b)$;
    
    \item $\determined^\psi(x_i) \implication \determined^\phi(x_i)$;
    
    \item $\deepirrelevant^\psi(x_i) \implication \deepirrelevant^\phi(x_i)$;
    
    \item $\dependent^\psi(V, x_i)  \implication \dependent^\phi(V, x_i) $.
    
    \end{itemize}
    \label{prop:quant-properties-are-stronger}
  \end{proposition}
  
  We note that the idea of relaxing universal quantifiers and approximating a QCSP 
  by a classical, existential CSP, has been considered implicitly 
  by several authors: the solver presented in 
  \cite{Benedetti-Lallouet-Vautard:IJCAI:2007} is built on top of a classical CP
  solver and its propagation mechanism essentially relies on the classical notion of
  inconsistency; other authors 
  \cite{Mamoulis-Stergiou:REPORT:2004,Gent-Nightingale-Stergiou:IJCAI:2005} have 
  investigated the use of substitutability in QCSP; here again the notion they have used
  was essentially the classical, existential one.
  
  Replacing a universal quantifier by an existential one is but one
  way to obtain a \emph{relaxation} of a QCSP. In \cite{Ferguson-OSullivan:IJCAI:2007},
  a more comprehensive list of relaxation techniques is studied. Interestingly this
  work essentially defines a relaxation as a transformation that guarantees that 
  \emph{if the relaxation is false, then so is the original problem}. In other words,
  the notion of relaxation is based on the truth of the QCSP. Proposition 
  \ref{prop:quant-properties-are-stronger} shows that \emph{quantifier relaxation}
  provides a way to do approximate reasoning on other properties than \emph{truth}.

  The \emph{incompleteness} of the reasoning on the existential relaxation is
  easily seen on an example:
  
  \setcounter{myexamplecounter}{1}
  \begin{example} (Ctd.)
    Consider the QCSP:
    \[
    \exists x_1 \in [2,3] ~ \forall x_2 \in [3,4]
    ~\exists x_3 \in [3,6].  ~x_1 + x_2 \leq x_3
    \]
    (See Fig. \ref{fig:killer-example}.)
    Noticeable properties are: 
    $\inconsistent(x_1, 3) $, 
    $\implied(x_1, 2)$,

    \noindent    
    $\deepfixable(x_1, 2)$, 
    $\deepremovable(x_1, 3)$,
    $\deepsubstitutable(x_1, 3, 2)$,
    $\determined(x_1)$.

    On the contrary if we apply the classical definition or, equivalently,
    consider the CSP 
    $
    \exists x_1 \in [2,3]~ \exists x_2 \in [3,4]~
    \exists x_3 \in [3,6].~ x_1 + x_2 \leq x_3
    $, \emph{none} of the properties holds, because of the tuple 
    $\langle 3, 3, 6 \rangle$. 
  \end{example}
  
  This confirms that the properties we
  have defined are new notions which do make a difference compared to
  classical CSP notions, and which allow a finer reasoning taking into
  account the quantifier prefix as well as the constraints themselves.

  \subsubsection{Shallow properties \vs\ deep properties}

  To complete the picture, we have the following relations between
  deep and shallow notions (the deep ones are more restrictive):
  \begin{proposition}
    For all variables $x_i$ and values $a$ and $b$, we have:
    
    \begin{itemize}
    
    \item $\deepfixable(x_i, a) \implication \shallowfixable(x_i, a)$;
    
    \item $\deepremovable(x_i, a) \implication \shallowremovable(x_i, a)$;
    
    \item $\deepsubstitutable(x_i, a, b) \implication \shallowsubstitutable(x_i, a, b)$;

    \item $\deepinterchangeable(x_i, a, b) \implication \shallowinterchangeable(x_i, a, b)$;
    
    \item $\deepirrelevant(x_i) \implication \shallowirrelevant(x_i)$.
    
    \end{itemize}
    \label{prop-shallow-stronger-than-deep}
  \end{proposition}
  
  Note that whether a property holds is always dependent on
  the quantification order. In the case of shallow definitions, this is 
  even more true, because the ordering matters even within a block of
  variables \emph{of the same nature}, for instance when the quantifiers are
  all existential. To see that, consider the QCSP:
  \[
  \exists x_1 \in [1,2]
  ~ \exists x_2 \in [3,4] ~\exists x_3 \in [4,6].  ~x_1 + x_2 =
  x_3.
  \]
  Value 1 is shallow-substitutable to 2 for $x_1$, and $x_1$ is 
  shallow-irrelevant, while 1 is not deep-substitutable to 2 for $x_1$
  (\ie\ substitutable in the classical sense), nor is $x_1$ deep-irrelevant.
  The intuition behind this is that here we consider that $x_1$ is
  assigned first, and \emph{at this step} the two choices are
  equivalent. In other words, the property holds
  \emph{because we are considering the ordering $x_1, x_2, x_3$}. 
  
  Interestingly, shallow properties, and shallow substitutability in 
  particular, provide a new, general form of properties even for the case
  of classical CSP. These properties are more general because they take into
  account information on a particular variable ordering. An interesting
  question is to determine the variable ordering that allows to detect
  the highest number of substitutability properties in a given CSP.

  \subsection{Relations between Properties}
  
  As in the classical case \cite{Bordeaux-Cadoli-Mancini:LPAR:2004}, 
  we also have relations between the properties, for instance a
  value that is implied is also deep-fixable (and therefore also shallow fixable);
  a variable that is (deep/shallow) irrelevant is also (deep/shallow) fixable
  to any value, \etc\ 
  We list the most remarkable of these relations in the next proposition:

  \begin{proposition}
    The following relations hold between the properties (forall $x_i$,
    $a$ and $b$):
    \begin{itemize}
    
    \item $\inconsistent(x_i, a) \implication 
      \forall b \in D_{x_i}. ~\deepsubstitutable(x_i, a ,b)$;

    \item $\implied(x_i,a) \equivalence
      \forall b \in D_{x_i} \setminus \{a\}. ~\inconsistent(x_i,b)$;

    \item $\implied(x_i,a) \implication
      \deepfixable(x_i, a)$;

    \item $\inconsistent(x_i, a) \implication 
      \deepremovable(x_i,a)$;

    \item $\exists b \in D_{x_i} \setminus \{a\}. ~\deepsubstitutable(x_i,a,b)
      \implication 
      \deepremovable(x_i, a)$;

    \item $\exists b \in D_{x_i} \setminus \{a\}. ~\shallowsubstitutable(x_i,a,b)
      \implication 
      \shallowremovable(x_i, a)$;

    \item $\deepfixable(x_i, b) \equivalence
      \forall a \in D_{x_i}. ~\deepsubstitutable(x_i, a, b)$;

    \item $\shallowfixable(x_i, b) \equivalence
      \forall a \in D_{x_i}. ~\shallowsubstitutable(x_i, a, b)$;

    \item $\deepirrelevant(x_i) \equivalence
      \forall a \in D_{x_i}. ~\deepfixable(x_i, a)$;

    \item $\shallowirrelevant(x_i) \equivalence
      \forall a \in D_{x_i}. ~\shallowfixable(x_i, a)$.

    \end{itemize}
    \label{prop-classifications}
  \end{proposition}

\section{Simplifications Allowed When the Properties Hold} 
\label{sec:simplifications}
%  \subsection{Correctness of the Definitions}

  The goal of reasoning on the properties of a QCSP is typically to 
  simplify the problem. In the cases we are interested in, this can
  be done in two ways: (1) by removing an element from the list of values 
  to consider for one of the variables, or (2) by instantiating a variable to a 
  particular value. Such simplifications are helpful for backtrack search 
  algorithms, which are typically considered when solving QCSP.
  
  We now show that the properties we defined allow simplifications
  that are \emph{correct}, in the sense that they do not alter the truth of the QCSP:
%  We say that the definitions are \emph{correct} if these simplifications do not alter the 
%  truth of the QCSP. We now check that our definitions are correct in this very
%  sense, more precisely: 

  \begin{itemize}
  
  \item
  If a value is removable for a given variable, then removing the
  value from the domain of that variable does not change the truth of the problem.
  
  \item
  If a value is fixable to a particular value for a given variable, 
  then instantiating the variable to this value does not change
  the truth of the problem.
  
  \end{itemize}
  
  The interest of the other properties lies essentially in their relation with the 
  two fundamental properties of removability and fixability, as expressed by
  Prop. \ref{prop-classifications}. For instance, an implied value is of
  interest essentially because it is fixable, and an irrelevant variable is
  of interest essentially because it is fixable to any value of
  its domain. Similarly, the interest of, e.g., inconsistent and substitutable 
  values is that they are removable. We therefore focus on proving the
  correctness of the two notions of removability and  fixability, and we will 
  consider their shallow forms: 
  recall that, by Prop. \ref{prop-shallow-stronger-than-deep}, the shallow are
  the stronger ones; a value which is deep-removable or deep-fixable is also
  shallow-removable or shallow-fixable, respectively.

  \subsection{Simplifying Existental Variables}

  Our whole game-theoretic approach is
  naturally biased towards existential variables: the notion of strategy
  considers that the values for the universal variables can be arbitrary, and
  specifies the values that should be taken for the existential ones.
  As a consequence, the approach is more naturally fitted to make deductions
  on the existential variables, and we first focus on this case.

  The simplifications allowed for an existential variable 
  when the removability property holds rely on the following Proposition:

  \begin{proposition}
    Let $\phi = \langle X, Q, D, C\rangle$ be a QCSP in which value $a
    \in D_{x_i}$ is shallow-removable for an existential variable $x_i$, 
    and let $\phi'$ denote the
    same QCSP in which value $a$ is effectively removed (\ie\ $\phi' =
    \langle X, Q, D', C\rangle$ where $D'_{x_i} = D_{x_i} \setminus
    \{a\}$ and $D'_{x_j} = D_{x_j}, \forall j \not= i$). Then $\phi$ is true
    iff $\phi'$ is true.
    \label{prop:correctness-of-removability}
  \end{proposition}

  The simplifications allowed for an existential variable 
  when the fixability property holds rely on the following Proposition:

  \begin{proposition}
    Let $\phi = \langle X, Q, D, C\rangle$ be a QCSP in which value $a
    \in D_{x_i}$ is shallow-fixable for an existential variable $x_i$, 
    and let $\phi'$ denote the same
    QCSP in which value $a$ is effectively fixed (\ie\ $\phi' =
    \langle X, Q, D', C\rangle$ where $D'_{x_i} = \{a\}$ and $D'_{x_j} =
    D_{x_j}, \forall j \not= i$). Then $\phi$ is true iff $\phi'$ is
    true.
    \label{prop:correctness-of-fix}
  \end{proposition}

  \subsection{Simplifying Universal Variables}

  To allow a proper, symmetric treatment of all variables of QCSPs
  it is necessary to also define how to make deductions on universal variables.
  The way this can be done has been suggested by several authors in the literature
  and is developed, for instance, in  \cite{Bordeaux-Zhang:SAC:2007}: 
  to make deductions on the universal variables, which represent the 
  ``moves of the opponent'', we have to reason on the negation of the 
  formula, which captures the ``winning strategies of the opponent''. 

  We say that a value is \emph{dual-}shallow-removable if it is 
  shallow-removable in the negation of the considered QCSP, and that it is
  \emph{dual}-shallow-fixable if it is shallow-fixable in this negation.
  The simplifications allowed for a universal variable 
  when the removability property holds rely on the following Proposition:

  \begin{proposition}
    Let $\phi = \langle X, Q, D, C\rangle$ be a QCSP in which value $a
    \in D_{x_i}$ is dual-shallow-removable for a universal variable $x_i$, 
    and let $\phi'$ denote the
    same QCSP in which value $a$ is effectively removed (\ie\ $\phi' =
    \langle X, Q, D', C\rangle$ where $D'_{x_i} = D_{x_i} \setminus
    \{a\}$ and $D'_{x_j} = D_{x_j}, \forall j \not= i$). Then $\phi$ is true
    iff $\phi'$ is true.
    \label{prop:correctness-of-dual-removability}
  \end{proposition}

  The simplifications allowed for a universal variable 
  when the fixability property holds rely on the following Proposition:

  \begin{proposition}
    Let $\phi = \langle X, Q, D, C\rangle$ be a QCSP in which value $a
    \in D_{x_i}$ is dual-shallow-fixable for an universal variable $x_i$, 
    and let $\phi'$ denote the same
    QCSP in which value $a$ is effectively fixed (\ie\ $\phi' =
    \langle X, Q, D', C\rangle$ where $D'_{x_i} = \{a\}$ and $D'_{x_j} =
    D_{x_j}, \forall j \not= i$). Then $\phi$ is true iff $\phi'$ is
    true.
    \label{prop:correctness-of-dual-fix}
  \end{proposition}

\section{Complexity results}
\label{sec:complexity}

  In this section, we study the complexity of the problem of determining
  whether the properties defined in Definitions~\ref{def:deep-properties}
  and~\ref{def:shallowProperties} hold. As was to be expected, our 
  results show that the problem  is in general intractable, 
  and we essentially obtain PSPACE-completeness results. In other words
  the complexity of checking one of the properties is typically the same
  as the complexity of determining whether the QCSP is true 
  \cite{Papadimitriou:BOOK:1994,Stockmeyer-Meyer:STOC:1973}.

  \subsection{Encoding Issues}
  \label{subsec-encoding-issues}
  
  To analyze the complexity, a few words are needed on the encoding of the QCSP
  $\langle X, Q, D, C\rangle$. Def. \ref{definition-QCSP} did not specify
  anything on this issue, because the encoding did not have any consequence on 
  the results of previous sections. We assume that $X$ and $Q$
  are encoded in the natural way, \ie\ as a list. For the set of domains $D$, two 
  choices may be considered: a domain can be encoded as a list of
  allowed values or as an interval, in which case its two bounds need to be
  encoded. Our results will hold independently of whether the interval
  or domain representation is chosen. The main question is how the
  constraints are defined. Some examples of representation 
  formalisms are the following:

  \begin{enumerate}

  \item[I] The domain is Boolean, \ie\ $B = \{0,1\}$, and $C$ is defined 
    as a Boolean circuit. 
    
  \item[II] The domain is Boolean, \ie\ $B = \{0,1\}$, and $C$ is put
    in Conjunctive Normal Form, \ie\ it is a conjunction of clauses 
    (disjunctions of literals, each of which
    is a variable or its negation).
  
  \item[III] $C$ is a conjunction of constraints, each of which is 
    represented in extension as a table (\eg\ binary) which 
    lists all tuples that are accepted.

  \item[IV] $C$ is a conjunction of constraints, each of which
    is represented by a numerical (linear or polynomial) equality or
    inequality.

  \item[V] $C$ is a polynomial-time \emph{program} (written in any universal
    language, for instance the Turing machine) which, given a tuple $t$,
    determines whether $t \in \sol$. 
  
  \end{enumerate}

  In all cases we impose the restriction that testing whether $t \in \sol$ 
  be feasible in polynomial time. The fifth encoding represents the most general
  possible encoding satisfying this restriction: we shall consider it when we want to 
  check that a result holds for any encoding in which testing whether $t \in \sol$ 
  can be done in polynomial time.
  
  Using encoding (V) to capture the notion of ``most general encoding'' is therefore
  convenient, but an important point is that the 4 other formalisms are essentially 
  as concise as formalism (V). If the domain is Boolean, then if $\sol$ can be
  represented by a program $P$ (in the sense that $P(t) = 1$ iff $t \in \sol$)
  and if the execution of $P$ requires a memory bounded by $S$ and a time bounded by $L$, 
  then the set $\sol$ can be also represented by a Boolean circuit of size polynomial
  in $S$, $L$, and the length of the text of the program $P$, using the technique 
  used by Cook in proving that SAT is NP-complete. In other words, for Boolean domains,
  formalism (I) is as expressive as formalism (V). 
  Now the relations between formalism
  (I) and formalisms (II) to (IV) are well-known: we can reduce a circuit to a CNF 
  involving only clauses of size at most three (3CNF)
  by introducing existential variables, and it is straightforward to reduce a 3CNF
  to formalism (III) or formalism (IV).
  The complexities of our problems for (I) to (V) will therefore be
  equivalent except for minor refinements occurring at intermediate levels of the 
  polynomial hierarchy (Prop \ref{prop-ph-levels}), where introducing existential
  variables makes a little difference.

  \subsection{A Common Upper Bound: PSPACE}

  The most difficult side of our complexity characterizations is to prove 
  \emph{membership} in PSPACE. It is indeed not completely obvious at first
  that the properties we have studied can be verified in polynomial space.
  The key point is to notice that a polynomial space algorithm exists 
  to recognize the set of outcomes. Considering representation (V), we
  have the following:
    
  \begin{proposition}
    Let $\phi = \langle X, Q, D, C\rangle$ be a QCSP.
    Given a tuple $t \in \prod_{x \in X} D_x$,
    we denote by $B$ the conjunction of constraints:
    \begin{equation}
    \bigwedge_{x_i \in E}
    \left(
      \left(\bigwedge_{y \in A_{i-1}} y = t_y\right) 
      \rightarrow (x_i = t_{x_i})
    \right) 
    \label{eq-additional-constraints-for-outcome}
    \end{equation}
    The QCSP $\psi = \langle X, Q, D, B \cup C\rangle$ is true
    iff $t \in \out^\phi$.
    \label{prop-trick-for-pspace-upper-bound}
  \end{proposition}

  Note that $B \cup C$ can be expressed concisely
  in formalism (V). The conjunction of constraints added in~\eqref{eq-additional-constraints-for-outcome}
  makes sure that
  any winning strategy of $\psi$ contains $t$ as a scenario.

  A direct corollary of Prop. \ref{prop-trick-for-pspace-upper-bound} 
  is that checking whether a particular tuple $t$
  belongs to the set of outcomes of a QCSP $\phi$ can be done in polynomial
  space, simply by solving $\psi$. This is true for any representation 
  of the constraints that respects the restriction that testing whether 
  $t \in \sol$ be feasible in polynomial time\footnote{In fact this condition 
  could itself be considerably relaxed: the PSPACE membership result holds under the
  very general condition that testing whether $t \in \sol$ be
  feasible in polynomial \emph{space}.}. Now being able to test in polynomial space
  whether a tuple is an outcome, the membership in PSPACE of all properties
  becomes clear: for instance if we consider inconsistency
  ($\forall t \in \out. ~ t_{x_i} \not= a$) we can enumerate
  all tuples in lexicographical order, determine whether each of them is
  an outcome, and whether it satisfies the implication 
  $t \in \out. ~ t_{x_i} \not= a$. The precise list of results will be given
  in the next section, where we state completeness results (including both hardness and
  membership for the considered class).

  \setcounter{myexamplecounter}{3}
  \begin{example}
    Let us illustrate the idea of Prop. \ref{prop-trick-for-pspace-upper-bound} 
    on a simple example. Consider the QCSP
    $\exists x_1. ~ \forall y_1. ~ \exists x_2.~ \forall y_2. ~ \exists x_3. ~ C$, 
    where the domain of each variable
    is, for instance $\{0,1\}$. We want to determine whether the tuple
    $\langle x_1=0, y_1=0, x_2=0, y_2=0, x_3=0 \rangle$ is an outcome of the QCSP. This can be done
    by solving the QCSP in which the constraints of~\eqref{eq-additional-constraints-for-outcome}
    are added:
    \[
    \exists x_1.  \forall y_1.  \exists x_2. \forall y_2.  \exists x_3. ~ C \land 
    (x_1=0 \wedge (y_1=0\rightarrow x_2=0) \wedge ((y_1=0\wedge y_2=0)\rightarrow x_3=0)) 
    .\]
  \end{example}

  It might be useful to mention a possible source of confusion: it is the case
  that our PSPACE membership results hold for formalism (4), since it respects our
  restriction. This is true even if the domains $D_x$ are represented by intervals:
  even though an interval whose bounds are $n$-bit integers represents in general a  
  set of values of cardinality exponential in $n$, we can always iterate on these
  values using polynomial space.
  This should be contrasted with classical complexity results related to 
  arithmetics: in general deciding the truth of quantified linear constraints
  is extremely complex (hard for \textsf{NDTIME($2^{2^n}$)} by the 
  Fischer-Rabin theorem \cite{Fischer-Rabin:CHAPTER:1974}, 
  and therefore provably not in PSPACE $\subseteq$ EXPTIME), and if we
  consider quantified polynomial constraints the problem becomes undecidable
  (G\"{o}del's theorem). The key point is that in these cases the values of the
  variables can grow extremely large; as long as we bound the domains explicitly
  this problem does not arise, which is why we remain within PSPACE.

  \subsection{Complexity Characterizations}

  We now list the complexity results we obtain. These results hold
  for any of the 5 representations we have mentioned.

  \begin{proposition}
    Given a QCSP $\phi = \langle X, Q, D, C \rangle$, the problems of
    deciding whether:
    
    \begin{itemize}
    
    \item value $a \in D_{x_i}$ is \deepfixable, \deepremovable,
    \inconsistent, \implied\ for variable $x_i \in X$,
    
    \item value $a \in D_{x_i}$ is \deepsubstitutable\ to or
    \deepinterchangeable\ with $b \in D_{x_i}$ for variable $x_i \in X$,
    
    \item variable $x_i \in X$ is \dependent\ on variables $V \subseteq X$,
    or is \deepirrelevant\ 
    
    \end{itemize}

    \noindent
    are PSPACE-complete.
    \label{prop:dproperties-PSPACEc}  
  \end{proposition}

  \noindent
  An analogous result holds for the shallow properties:

  \begin{proposition}
  
    Given a QCSP $\phi = \langle X, Q, D, C \rangle$ , the problems of
    deciding whether:
    
    \begin{itemize}
    \item value $a \in D_{x_i}$ is \shallowfixable, \shallowremovable\ for
    variable $x_i \in X$,

    \item value $a \in D_{x_i}$ is \shallowsubstitutable\ to or
    \shallowinterchangeable\ with $b \in D_{x_i}$ for variable $x_i \in X$,

    \item variable $x_i \in X$ is \shallowirrelevant\ 
    
    \end{itemize}

    \noindent
    are PSPACE-complete.
    \label{prop:sproperties-PSPACEc}
  \end{proposition}

  As usual when considering quantified constraints, 
  the complexity increases with the number of quantifier alternations, more
  precisely each additional alternation brings us one level higher in the Polynomial
  Hierarchy \cite{Stockmeyer:TCS:1976}. The precise level that is reached is
  dependent on the considered property and
  on many details, including the formalism used for the encoding of the QCSP. 
  We shall not list all results but instead we characterize, as an example, the complexity
  obtained in a particular setting, \ie\ for the ``deep'' definitions of the properties,
  in the case where the QCSP starts with an existential quantifiers, and where
  its constraints are encoded as a Boolean circuit.
  
  We call $\Sigma_k$QCSPs the QCSPs with at most $k$ quantifier alternations  
  and whose first variables are existential. We have the following results:

  \begin{proposition}
    Given a $\Sigma_k$QCSP $\phi = \langle X, Q, D, C \rangle$ encoded using
    Formalism (I), the problems of deciding whether:
    
    \begin{itemize}
    
    \item value $a \in D_{x_i}$ is \textsl{deep-fixable},
    \textsl{deep-removable}, \inconsistent, \implied\ for
    variable $x_i \in X$,
    
    \item value $a \in D_{x_i}$ is \textsl{deep-substitutable} to or
    \textsl{deep-interchangeable} with $b \in D_{x_i}$ for variable
    $x_i \in X$,
    
    \item variable $x_i \in X$ is \dependent\ on variables $V \subseteq X$,
    or is \textsl{deep-irrelevant},
    
    \end{itemize}

    \noindent    
    are $\Pi_k^p$-hard and belong to $\Pi_{k+1}^p$. Moreover, for deep inconsistency, 
    implication, determinacy and dependence, the problems are more precisely 
    $\Pi_k^p$-complete.
    \label{prop:bounded:deepAndShallowProperties}
    \label{prop-ph-levels}
  \end{proposition}
  
  In particular, it was reported in 
  \cite{Bordeaux-Cadoli-Mancini:LPAR:2004} that these problems are
  coNP-complete for purely existential QCSPs. 
  
  Why the precise results are less regular than in previous cases is 
  because the precise number of quantifier alternations is impacted by 
  many factors. 
  For instance, if we consider a Quantified
  Boolean Formula $\exists X. ~\forall Y. ~ F(X, Y)$, where $X$ and $Y$ are vectors
  of Boolean variables and $F$ is a Boolean circuit, then putting $F$ into CNF
  will produce a formula of the form $\exists X. ~\forall Y. ~\exists Z. ~ G(X, Y, Z)$,
  and this sometimes incurs a difference of one level in the polynomial hierarchy
  between Formalism (I) and Formalisms (II) to (IV). 
  Similarly, there is a difference between shallow and deep properties in that 
  shallow properties are themselves usually stated with more quantifier
  alternations, a typical form being ``forall outcomes, there exists an outcome''.  
  What is obviously true for all properties in any case, however, is if we consider
  QCSPs with a limited number of quantifier alternations, the level reached in the
  polynomial hierarchy is also bounded.

\section{Local reasoning}
\label{sec:locality}

  The previous section shows that all of the properties we are
  interested in are computationally difficult to detect---in fact
  as difficult as the resolution of the QCSP problem itself. 
  There are nonetheless particular cases where a property can be cheaply 
  revealed. In CSP solvers the most widely used way of detecting properties
  cheaply is by using \emph{local} reasoning: instead of analysing the
  whole problem at once, thereby facing its full complexity, we analyse it
  bit by bit (typically constraint by constraint). Depending on the property
  we know how deductions made on the bits generalize to the whole QCSP.
  For instance:
  
  \begin{itemize}
  
  \item In the case of inconsistency, a deduction made on one single constraint
  generalizes to the whole CSP. For instance, if we have a CSP
  $\exists x \in [0,5]. ~  y \in [0,5]. ~ x > y \wedge C$, we can deduce from
  the constraint $x > y$ that value 0 is inconsistent for $x$, without having to
  worry of which other constraints are present in $C$.
  
  \item In the case of substitutability, a deduction is valid for the whole QCSP if
  it can be checked independently for each and every constraint. For instance
  if we have the CSP $\exists x \in [0,5]. ~  y \in [0,5]. ~ x > 1 \wedge x \leq y$, we
  can deduce that value value 3 is substitutable to 2 for $x$. This is the case because
  the substitutability property holds for both constraints $x > 1$ and $x \leq y$.
  If, however, there were a third constraint, we would have to make sure that the property holds
  for it as well before deducing that it holds for the whole CSP. The
  situation is slightly less advantageous than for inconsistency because we have to 
  consider each constraint before making a deduction, but it is nevertheless of
  interest---analysing the constraints one by one is typically much cheaper than
  analysing the whole CSP at once.
  
  \end{itemize}
  
  Following the classical CSP approach, we investigate the use of
  local reasoning as a means to cheaply detect the properties we have proposed.

  \subsection{Positive Results}

  Our first result is that using local reasoning allows to detect the
  deep properties except removability. 
  Depending on the property one of the two forms of
  generalization mentioned before is correct.

  \begin{proposition}
    Let $\phi = \langle X, Q, D, C \rangle$ be a QCSP where $C =
    \{c_1, \dots, c_m\}$. We denote by $\phi_{k}$ the QCSP $\langle
    X, Q, D, \{c_k\} \rangle$ in which only the $k$-th constraint is
    considered. We have, for all $x_i \in X$, $V \subseteq X$, 
    and $a, b \in D_{x_i}$:

    \begin{itemize}
    
    \item $\betweenparenth{\bigvee_{k \in 1 .. m} \inconsistent^{\phi_k}(x_i, a)}
      \implication  \inconsistent^{\phi}(x_i, a)$;

    \item $\betweenparenth{\bigvee_{k \in 1 .. m} \implied^{\phi_k}(x_i, a)}
      \implication  \implied^{\phi}(x_i, a)$;
      
    \item $\betweenparenth{\bigwedge_{k \in 1 .. m} \deepfixable^{\phi_k}(x_i, a)}
      \implication  \deepfixable^{\phi}(x_i, a)$;
      
    \item $\betweenparenth{\bigwedge_{k \in 1 .. m} \deepsubstitutable^{\phi_k}(x_i, a, b)}
      \implication  \deepsubstitutable^{\phi}(x_i, a, b)$;
      
    \item $\betweenparenth{\bigwedge_{k \in 1 .. m} \deepinterchangeable^{\phi_k}(x_i, a, b)}
      \implication  \deepinterchangeable^{\phi}(x_i, a, b)$;
      
    \item $\betweenparenth{\bigvee_{k \in 1 .. m} \determined^{\phi_k}(x_i)}
      \implication  \determined^{\phi}(x_i)$;
      
    \item $\betweenparenth{\bigwedge_{k \in 1 .. m} \deepirrelevant^{\phi_k}(x_i)}
      \implication  \deepirrelevant^{\phi}(x_i)$;
      
    \item $\betweenparenth{\bigvee_{k \in 1 .. m} \dependent^{\phi_k}(V, x_i)}
      \implication  \dependent^{\phi}(V, x_i)$.
      
    \end{itemize} 
  \end{proposition}

%\toni{REPETITIVE: could be removed}
%  Local reasoning can therefore be used to check that some property
%  holds by inspecting the constraints one by one without considering
%  the problem as a whole, for instance that a value is substitutable
%  to another just because this property holds for each constraint. In
%  some other cases like inconsistency, it also allows to determine
%  that a property holds just because one particular constraint has the
%  property.

  \subsection{Negative Results}

  It was noticed in \cite{Bordeaux-Cadoli-Mancini:LPAR:2004} that, even in 
  the non-quantified case, deep removability is not as well-behaved as the other deep
  properties since it is not possible to detect it using local reasoning. 
  This was seen on an example, which we borrow from this paper:
  
  \begin{example}
    Consider the CSP 
    \[
    \exists x \in \{1,2,3\}. ~ \exists y \in \{1,2,3\}. ~
    (x \leq y, y \leq x, x \not= 1, x \not= 3)
    \]
    If we consider each of the four constraints, then we find that
    value 2 is removable for $x$. But obviously value 2 is \emph{not}
    removable for the CSP as the only solution is indeed $x=2, y=2$.
  \end{example}

  A similar problem occurs when we consider the shallow definitions: it is incorrect,
  in general, to use local reasoning to detect these versions of the properties\footnote{%
    This corrects an error in \cite{Bordeaux-Cadoli-Mancini:AAAI:2005}, where we
    wrongly stated that local reasoning is valid for all properties.
  }.
  Here again this can be seen on a simple example:
  
  \begin{example}
    Consider the (Q)CSP
    \[
    \exists x_1 \in \{0,1\}. ~ \exists x_2 \in \{0,1\}. ~ 
    (x_1 = x_2 \wedge x_2 = 1)
    \]
    It is the case that variable $x_1$ is shallow-fixable to value 0 \wrt\ constraint $x_1=x_2$;
    and variable $x_1$ is also shallow-fixable to value 0 \wrt\ constraint $x_2=1$. 
    Despite of that, $x_1$ is not shallow-fixable to 0 in the QCSP, as there is simply no solution
    with $x_1=0$. 
  \end{example}
  
  The shallow definitions therefore have to be considered carefully: they are more general than the
  deep properties, but they have to be detected by other means than local reasoning.
  This is somewhat reminiscent of what happens with the
  removability property, whose generality comes at the price of being a less well-behaved
  property than substitutability or inconsistency.

\section{Concluding Remarks}
\label{sec:conclusion}

  \subsection{Related Works}

  A number of works related to Quantified CSP have considered particular
  cases of the properties we have attempted to study systematically in this 
  paper. Most of these works have been mentioned throughout the paper, 
  notably \cite{Mamoulis-Stergiou:REPORT:2004} for their use of substitutability; we 
  also note the work done by Peter Nightingale in his thesis, which 
  devotes large parts to the consistency property \cite{Nightingale:CP:2005}.
  The notions considered in these works are related to our proposals but typically 
  less general, because our definitions finely take into account the quantifiers. 
  For substitutability for instance, the definition used in \cite{Mamoulis-Stergiou:REPORT:2004}
  was essentially the classical (existential) definition.
  For consistency, our definition subsumes the notions proposed by
  \cite{Bordeaux-Monfroy:CP:2002} or \cite{Nightingale:CP:2005}. 
  Our general definition nevertheless leaves open the question of how to efficiently detect
  inconsistent values, and these proposals can be seen as  
  particular ways of using local reasoning to detect inconsistent values. This situation is
  quite closely related to works in CSP, where many notions of local consistency can be 
  defined. These notions have different merits that can be evaluated experimentally,
  but they all share the basic property of being ways to detect (globally) inconsistent 
  values, which explains why they are correct. 

  We also note that more advanced studies are available for the particular case of
  Boolean quantified constraints. In these works some techniques have been proposed
  that specifically take into account the quantifier prefix. 
  However, contrary to ours, these proposals are restricted to Boolean domains. 
  For instance in \cite{Rintanen:IJCAI:1999,Cadoli-Schaerf-Giovanardi-Giovanardi:JAR:2002},
  several techniques are proposed to fix and remove values.  These
  works have shown that detecting properties is essential and can lead
  to a consistent pruning of the search space, but no clear and
  general framework to understand these properties was available. 

  An interesting, recent related
  work is  \cite{Audemard-Jabbour-Sais:IJCAI:2007}, which initiates the study 
  of \emph{symmetries} in Quantified Boolean Formulae.
  Symmetries are related to the notion of interchangeability but are in a sense
  a more general concept. Our feeling is that the idea of using the 
  notion of outcome to define constraint properties may be applicable to
  this class of properties as well. Symmetries are a complex and fascinating topic; 
  an interesting perspective for future work will be to see if our framework 
  can help understanding them in the general context of quantified CSP.
  
  \subsection{Conclusion}

  A primary goal of our work was to state the definitions in a way that
  is formal and amenable to proofs. In previous QCSP literature, it is fair to say that
  formal proofs were scarce, probably because facts that are trivial to prove in CSP
  tend to become complex to write formally when quantifiers come into play.
  Quantifiers can be complex to reason with, and it is sometimes easy to make 
  wrong assumptions on some properties, as we saw ourselves when finding the error we 
  made in the preliminary version of this paper (Section 6). 
  Because of this difficulty, we wanted in this work to build solid foundations on which
  the deductions made in QCSP solvers can rely.

\bibliographystyle{acmtrans}
\bibliography{biblio}

\begin{received}
  Received XXXXX;
  accepted XXXXX
\end{received}

%%%%%%%%%%%%%%%%%%%%%%%%%%%%%%%%%%%%%%%%%%%%%%%%%%%%%%%%%%%%%%%%%%%%%%%%%%%%%%%
\elecappendix

  \section{Proofs of the Main Propositions}
  \label{app-proofs}

  \setcounter{mypropositioncounter}{0} 

  \begin{proposition}
    A QCSP is true  (as defined in Section \ref{subsec-truth}) 
    iff it has a winning strategy.
  \end{proposition}
  
    \begin{proof}
      Instead of proving this result from scratch we sketch its connection to 
      classical logical results and simply note that the functions used in 
      the definition of the notion of strategy are essentially Skolem functions: 
      it is well-known that, starting from a formula
      $
      \forall x_1 \dots x_n. \exists y. ~F(x_1, \dots x_n, y)
      $
      with an existentially quantified variable $y$, we can replace $y$ by
      a function and obtain a second-order formula that is equivalent:
      $
      \exists f. ~ \forall x_1 \dots x_n. ~F(x_1, \dots x_n, f(x_1 \dots x_n))
      $. 
      
      If the domain $\domain$ is additionally fixed and each quantifier is
      bounded, \ie\ if we have a formula of the form:
      $
      \forall x_1 \in D_{x_1} \dots \forall x_n \in D_{x_n}. \exists y \in D_y. ~F(x_1, \dots x_n, y)
      $,
      then the formula is equivalent to:
      \[
      \exists f. ~
      \forall x_1 \in D_{x_1} \dots \forall x_n \in D_{x_n}. 
      \left(
      f(x_1 \dots x_n) \in D_y\wedge
      ~F(x_1, \dots x_n, f(x_1 \dots x_n))\right)
      \]
      and any interpretation $I$ verifying:
      \[
      \langle \domain, I \rangle ~\models~ 
      \forall x_1 \in D_{x_1} \dots \forall x_n \in D_{x_n}.
      \left(
      f(x_1 \dots x_n) \in D_y\wedge
      ~F(x_1, \dots x_n, f(x_1 \dots x_n))\right)
      \]
      is such that the function $I(f)$ is of signature 
      $
      \left(\prod_{x_i \in \{x_1 \dots x_n\}} D_{x_i}\right) \rightarrow D_y.
      $
      
      Now given a QCSP, let $F$ be its logical representation as defined in Section
      \ref{subsec-truth}, and let $F'$ be the Skolem normal form of $F$,
      obtained by iteratively applying the process described above, for all existential
      variables. 
      The strategies of the QCSP are exactly the possible interpretations of the Skolem 
      functions of $F'$. Furthermore, a strategy is winning (all outcomes are true) iff
      the first-order (universally quantified) part of the formula is true. 
      Consequently a winning strategy exists for 
      the QCSP iff the model-checking problem $\langle \domain, I \rangle \models F'$ is true,
      \ie\ iff the QCSP is true.
    \end{proof}

  \begin{proposition}   
    Deep fixability could equivalently be defined by the condition
    $\forall t \in \out. t[x_i := a] \in {\sol}$;
    Deep substitutability could be equivalently defined by 
    $\forall t \in \out.$  $(t_{x_i} = a) \implication (t[x_i:=b] \in \sol)$;   
    deep removability by 
    $\forall t \in \out.
          (t_{x_i} = a)  \implication (\exists b \neq a.
          t[x_i := b] \in \sol)$;
    and deep irrelevance by
    $\forall t \in \out. \forall b \in D_{x_i}. ~ t[x_i := b] \in \sol$.
  \end{proposition}
    
    \begin{proof}
      We consider fixability and we prove that
      $
      \forall t \in \out. ~t[x_i := a] \in {\out}
      $
      holds iff
      $
      \forall t \in \out. ~t[x_i := a] \in \underline{\sol}
      $
      does.
      The $\rightarrow$ implication is straightforward ($\out \subseteq \sol$); 
      we prove the $\leftarrow$ implication. In the case where the QCSP is false
      (no winning strategy) the implication trivially holds, since
      \out\ is then empty. Let us therefore prove it in the case where 
      the QCSP is true. 

      We assume that $\forall t \in \out. ~t[x_i := a] \in \sol$. Let $t \in \out$; 
      it is clear that the tuple  $t[x_i := a]$ belongs to \sol;
      we have to prove that $t[x_i := a]$ also belongs to \out. For that purpose, we
      exhibit a winning strategy $s$ such that $t[x_i := a] \in \sce(s)$. 

      Let $s'$ be a winning strategy such that $t \in \sce(s')$. Such a strategy exists
      since $t$ is an outcome. The strategy $s$ will be obtained by modifying $s'$
      so that all its outcomes assign value $a$ to variable $x_i$.
      More formally, the functions $s_{x_j}$ are defined,
      for each $x_j \in E$, as follows:
      
      \begin{itemize}
      
      \item If $j = i$ then $s_{x_j}(\tau) \doteq a$, for each
         tuple $\tau \in \prod_{y \in A_{j-1}}D_y$;
      
      \item Otherwise $s_{x_j}$ is simply defined as the function $s'_{x_j}$.
      
      \end{itemize}
      
      One can now verify that $\sce(s) = \{\tau[x_i := a] ~:~ \tau \in \sce(s')\}$.
      Two consequences are $t[x_i := a] \in \sce(s)$, and $\sce(s) \subseteq \sol$, 
      which show that $s$ is a winning strategy such that $t[x_i := a] \in \sce(s)$.
      
      Similarly, for substitutability we can exhibit a strategy 
      $s$ in which every $t \in \sce(s')$ such that $t_{x_i} = a$ is changed into the
      scenario $t[x_i := b]$.
      
      For removability it is convenient to restate the property: removability holds if
      there exists a function $f$ that associates to every $X$-tuple $t$ a value 
      $f(t) \not= a$, and such that $\forall t \in \out.
          (t_{x_i} = a)  \implication (t[x_i := f(t)] \in \out)$. 
      We can exhibit a strategy $s$ in which every $t \in \sce(s)$ such that 
      $t_{x_i} = a$ is changed into the scenario $t[x_i := f(t)]$.
      
      For irrelevance we can use the fact that a variable is irrelevant iff it can
      be fixed to any value of its domain (Prop. \ref{prop-classifications}).
    \end{proof}

  \begin{proposition}
    Let $\phi = \langle X, Q, D, C\rangle$ be a QCSP and let 
    $\psi$ be the same QCSP but in which all quantifiers are existential,
    \ie\ $\psi = \langle X, Q', D, C\rangle$, with $Q'_{x} = \exists$, for all
    $x \in X$. We have (forall $x_i, a, b, V$):
    
    \begin{itemize}
    
    \item $\inconsistent^\psi(x_i, a) \implication \inconsistent^\phi(x_i, a)$;
    
    \item $\deepfixable^\psi(x_i, a) \implication \deepfixable^\phi(x_i, a)$;
    
    \item $\deepsubstitutable^\psi(x_i, a, b) \implication \deepsubstitutable^\phi(x_i, a, b)$;
    
    \item $\deepremovable^\psi(x_i, a) \implication \deepremovable^\phi(x_i, a)$;
    
    \item $\deepinterchangeable^\psi(x_i, a, b) \implication \deepinterchangeable^\phi(x_i, a, b)$;
    
    \item $\determined^\psi(x_i) \implication \determined^\phi(x_i)$;
    
    \item $\deepirrelevant^\psi(x_i) \implication \deepirrelevant^\phi(x_i)$;
    
    \item $\dependent^\psi(V, x_i)  \implication \dependent^\phi(V, x_i) $.
    
    \end{itemize}
  \end{proposition}
    
    \begin{proof}
      All the results rely essentially on the fact that $\out \subseteq \sol$.
      For the properties of inconsistency, implication, determinacy and dependence,
      the proof directly follows:
      classical inconsistency means that $\forall t \in \sol. ~ t_{x_i} \not= a$, 
      which implies the deep property  $\forall t \in \out. ~ t_{x_i} \not= a$;
%      similarly classical implication means that $\forall t \in \sol. ~ t_{x_i}= a$,
%      which implies the deep property $\forall t \in \out. ~ t_{x_i}= a$;
      classical determinacy means that 
      $\forall t \in \sol. ~ \forall b \not=t_{x_i}. ~ t[x_i := b] \not\in \sol$, which
      implies $\forall t \in \out. ~ \forall b \not=t_{x_i}. ~ t[x_i := b] \not\in \sol$,
      which implies the deep property
      $\forall t \in \out. ~ \forall b \not=t_{x_i}. ~ t[x_i := b] \not\in \out$.
      The cases of implication and dependence are similar.
    
      For the other properties we additionally use Proposition \ref{prop-no-outcome}:
      classical fixability means that $\forall t \in \sol. ~ t[x_i := a] \in \sol$. This
      implies $\forall t \in \out. ~ t[x_i := a] \in \sol$ which, by Proposition 
      \ref{prop-no-outcome}, is equivalent to the deep property 
      $\forall t \in \out. ~ t[x_i := a] \in \out$.
      The cases of substitutability, removability, interchangeability and irrelevance are 
      similar.
    \end{proof}

  \begin{proposition}
    For all variables $x_i$ and values $a$ and $b$, we have:
    
    \begin{itemize}
    
    \item $\deepfixable(x_i, a) \implication \shallowfixable(x_i, a)$;
    
    \item $\deepremovable(x_i, a) \implication \shallowremovable(x_i, a)$;
    
    \item $\deepsubstitutable(x_i, a, b) \implication \shallowsubstitutable(x_i, a, b)$;

    \item $\deepinterchangeable(x_i, a, b) \implication \shallowinterchangeable(x_i, a, b)$;
    
    \item $\deepirrelevant(x_i) \implication \shallowirrelevant(x_i)$.
    
    \end{itemize}
  \end{proposition}
  
    \begin{proof}
      If deep fixability holds, \ie\ we have $\forall t \in \out. ~~ t[x_i := a] \in \out$,
      then for each $t \in \out$ the tuple $t'=t[x_i := a]$ is such that 
      $t|_{X_{i-1}} = t'|_{X_{i-1}} \wedge t'_{x_i} = a$, and we therefore have
      $\forall t \in \out. ~
        \exists t' \in \out.  ~
        (t|_{X_{i-1}} = t'|_{X_{i-1}} \wedge  t'_{x_i} = a)$, which means
      $\shallowfixable(x_i, a)$. 
      The proof is similar for irrelevance.
      
      If deep removability holds, \ie\ 
      $\forall t \in \out. ~(t_{x_i} = a)  \implication (\exists b \neq a.
      ~t[x_i := b] \in \out)$, then for each $t \in \out$ such that 
      $t_{x_i} = a$, the tuple $t' = t[x_i := b]$ is such that 
      $t|_{X_{i-1}} = t'|_{X_{i-1}} \wedge t'_{x_i} = b$, and we have
      $\shallowremovable(x_i, a)$. The proof is similar for substitutability, 
      which also uses a bounded quantification, and the result follows for interchangeability.
    \end{proof}

  \begin{proposition}
    The following relations hold between the properties (forall $x_i$,
    $a$ and $b$):
    \begin{enumerate}
    
    \item $\inconsistent(x_i, a) \implication 
      \forall b \in D_{x_i}. ~\deepsubstitutable(x_i, a ,b)$;

    \item $\implied(x_i,a) \equivalence
      \forall b \in D_{x_i} \setminus \{a\}. ~\inconsistent(x_i,b)$;

    \item $\implied(x_i,a) \implication
      \deepfixable(x_i, a)$;

    \item $\inconsistent(x_i, a) \implication 
      \deepremovable(x_i,a)$;

    \item $\exists b \in D_{x_i} \setminus \{a\}. ~\deepsubstitutable(x_i,a,b)
      \implication 
      \deepremovable(x_i, a)$;

    \item $\exists b \in D_{x_i} \setminus \{a\}. ~\shallowsubstitutable(x_i,a,b)
      \implication 
      \shallowremovable(x_i, a)$;

    \item $\deepfixable(x_i, b) \equivalence
      \forall a \in D_{x_i}. ~\deepsubstitutable(x_i, a, b)$;

    \item $\shallowfixable(x_i, b) \equivalence
      \forall a \in D_{x_i}. ~\shallowsubstitutable(x_i, a, b)$;

    \item $\deepirrelevant(x_i) \equivalence
      \forall a \in D_{x_i}. ~\deepfixable(x_i, a)$;

    \item $\shallowirrelevant(x_i) \equivalence
      \forall a \in D_{x_i}. ~\shallowfixable(x_i, a)$.
    \end{enumerate}
  \end{proposition}
        
    \begin{proof}
      (1) Assume inconsistency holds.
      If we consider an arbitrary $t \in \out$, then $t_{x_i} \not= a$, which
      falsifies the left side of the implication 
      $(t_{x_i} = a) \implication (t[x_i:=b] \in \out)$, for any $b$,
      and deep substitutability therefore holds.
 
      (2) If value $a$ is implied for $x_i$, \ie\ $\forall t \in \out. ~ t_{x_i} = a$, 
      then for every value $b \not= a$ we have $\forall t \in \out. ~ t_{x_i} = a \not= b$, 
      \ie\ $b$ is inconsistent. If all values $b \not= a$ are inconsistent, \ie\ 
      $\forall t \in \out. ~ t_{x_i} \not= b$, then any $t \in \out$ is such that
      $\forall b \not= a. ~t_{x_i} \not= b$ and $t_{x_i} \in D_{x_i}$, so $t_{x_i} = a$\
      \ie\ $a$ is implied.

      (3) If $a$ is implied for $x_i$, then any $t \in \out$ is such that $t_{x_i} = a$, 
      and we therefore have $t[x_i := a] = t \in \out$.
      
      (4) If $a$ is inconsistent for $x_i$, \ie\ $\forall t \in \out. ~ t_{x_i} \not= a$,
      then the left-hand side of the implication 
      $(t_{x_i} = a)  \implication (\exists b \neq a.
          ~ t[x_i := b] \in \out)$ is false for every $t \in \out$.
          
      (5) If $a$ is deep-substitutable to a certain value $b \not= a$, then for every 
      $t \in \out$ verifying $t_{x_i} = a$ we have $t[x_i:=b] \in \out$.
      This implies $\exists b \neq a.~ t[x_i := b] \in \out$. 
      
      (6) If $a$ is shallow-substitutable to a certain value $b \not= a$, then for every  
      $t \in \out$ verifying $t_{x_i} = a$, we have
      $\exists t' \in \out. ~((t|_{X_{i-1}} = t'|_{X_{i-1}}) \wedge (t'_{x_i} = b))$.
      This implies $\exists t' \in \out. ~(t|_{X_{i-1}} = t'|_{X_{i-1}} \land t'_{x_i} \neq a)$.

      (7) If $b$ is deep-fixable for $x_i$, \ie\ $\forall t \in \out. ~ t[x_i := b] \in \out$,
      then the right-hand side of the implication $(t_{x_i} = a) \implication (t[x_i:=b] \in \out)$
      is true for all $t \in \out$.

      (8) If $b$ is shallow-fixable for $x_i$ \ie\ $\forall t \in \out. 
      ~\exists t' \in \out. ~(t|_{X_{i-1}} = t'|_{X_{i-1}} \wedge t'_{x_i} = b)$,
      then the right-hand side of the implication
      $t_{x_i} = a \implication \exists t' \in \out. ~((t|_{X_{i-1}} = t'|_{X_{i-1}})
      \wedge (t'_{x_i} = b)$ is true for all $t \in \out$.

      (9) If $x_i$ is deep-irrelevant, \ie\ 
      $\forall t \in \out. ~\forall a \in D_{x_i}. ~ t[x_i := a] \in \out$, then
      for any $a \in D_{x_i}$ we have $\forall t \in \out. ~ t[x_i := a] \in \out$.

      (10) If $x_i$ is shallow-irrelevant, \ie\ 
      $\forall t \in \out. ~\forall a \in D_{x_i}. ~\exists t' \in \out. ~
      (t|_{X_{i-1}} = t'|_{X_{i-1}}) \wedge (t'_{x_i} = a)$, then
      for any $a \in D_{x_i}$ we have
      $\forall t \in \out. 
      ~\exists t' \in \out. ~(t|_{X_{i-1}} = t'|_{X_{i-1}} \wedge t'_{x_i} = a)$.
    \end{proof}

  \begin{proposition}
    Let $\phi = \langle X, Q, D, C\rangle$ be a QCSP in which value $a
    \in D_{x_i}$ is shallow-removable for an existential variable $x_i$, and let $\phi'$ denote the
    same QCSP in which value $a$ is effectively removed (\ie\ $\phi' =
    \langle X, Q, D', C\rangle$ where $D'_{x_i} = D_{x_i} \setminus
    \{a\}$ and $D'_{x_j} = D_{x_j}, \forall j \not= i$). Then $\phi$ is true
    iff $\phi'$ is true.
  \end{proposition}
    
    \begin{proof} 
      If $\phi'$ has a winning strategy then the same strategy is also winning 
      for $\phi$; having $\phi'$ true therefore implies that $\phi$ is also true.
      
      On the other hand, assume that $\phi$ has
      a winning strategy $s^1$. Since $a \in D_{x_i}$ is shallow-removable for $x_i$,
      we have:
      \[
      \forall t \in \out. ~ t_{x_i} = a \implication \exists t' \in \out. ~
      (t|_{X_{i-1}} = t'|_{X_{i-1}} \land t'_{x_i} \neq a.)
      \]
      
      We show that if $s^1$ has a scenario $t \in \sce(s^1)$
      such that $t_{x_i} = a$, then we can ``correct'' this and exhibit another winning 
      strategy $s$ whose scenarios are the same as those of $s^1$ except 
      that all scenarios $\lambda$ such that $\lambda|_{X_{i-1}} = t|_{X_{i-1}}$
      have been replaced by tuples $t'$ with $t'_{x_i} \not= a$.
      (Intuitively we replace the ``sub-tree'' corresponding to the branch
      $t|_{X_{i-1}}$ by a new branch which does not involve the choice $x_i = a$ anymore.)
      More precisely, every scenario $t' \in \sce(s)$ will satisfy:
      \begin{itemize}
      
      \item If $t'|_{X_{i-1}} \not= t|_{X_{i-1}}$ then $t' \in \sce(s^1)$.
      
      \item If $t'|_{X_{i-1}} = t|_{X_{i-1}}$ then $t'_{x_i} \not= a$.
      
      \end{itemize}
      
      This will prove the result: in showing how to construct $s$ we show that,
      starting from any winning strategy $s^1$ containing a number $n > 0$ of
      ``incorrect'' scenarios $t'$ with $t'_{x_i} = a$, we can always exhibit a winning
      strategy with at most $n-1$ such scenarios, and repeating the correction $n$
      times we construct a 
      winning strategy in which no tuple $t'$ is such that $t'_{x_i}=a$.
      
      Let us now see how to construct $s$ starting from $s^1$. The outcome $t \in \sce(s^1)$
      that needs to be replaced is such that $t_{x_i} = a$ and, using the shallow 
      removability property, we conclude that there exists another outcome $\theta \in \out$ such that
      $\theta|_{X_{i-1}} = t|_{X_{i-1}} \wedge \theta_{x_i} \neq a$. This outcome belongs to
      at least one winning strategy. We choose one of these strategies, which we
      call $s^2$. To define the new strategy $s$ we must define the functions $s_{x_j}$, for
      each $x_j \in E$. These functions are defined as follows:
      
      \begin{itemize}
      
      \item if $j < i$ then $s_{x_j}$ is defined as $s^1_{x_j}$ (\eg\ we follow
      the strategy $s^1$ for the first variables, until variable $x_i$, excluded); 

      \item for the following variables, \ie\ when $j \geq i$, we define the value of 
      $s_{x_j}(\tau)$, for each $\tau \in \prod_{y \in A_{j-1}}D_y$, as follows:

        \begin{itemize}
        
        \item if $\tau|_{X_{i-1}} = t|_{X_{i-1}}$, then 
        $s_{x_j}(\tau) = s^2_{x_j}(\tau)$;
        
        \item if $\tau|_{X_{i-1}} \not= t|_{X_{i-1}}$, then 
        $s_{x_j}(\tau) = s^1_{x_j}(\tau)$;
        
        \end{itemize}
      
      \end{itemize}

      The proof is completed by checking that every scenario $t' \in \sce(s)$
      satisfies the two desired properties:

      \begin{itemize}
      
      \item If $t'|_{X_{i-1}} \not= t|_{X_{i-1}}$ then $t' \in \sce(s^1)$, because, for each
      $x_j \in E$, we have $t'_{x_j} = s_{x_j}(t'|_{A_{j-1}}) = s^1_{x_j}(t'|_{A_{j-1}})$ 
      in this case.
      
      \item If $t'|_{X_{i-1}} = t|_{X_{i-1}}$ then $t'_{x_i} \not= a$, because 
      $t'_{x_i} = s_{x_i}(t'|_{A_{i-1}}) = s^2_{x_i}(t'|_{A_{i-1}}) 
      = s^2_{x_i}(t|_{A_{i-1}})  = s^2_{x_i}(\theta|_{A_{i-1}}) = \theta_{x_i} \not= a$.
      
      \end{itemize}
      
      Furthermore, every $t' \in \sce(s)$ with $t'|_{X_{i-1}} = t|_{X_{i-1}}$
      belongs to $\sce(s^2)$, and $s$ is therefore a winning strategy:  
      $\sce(s) \subseteq (\sce(s^1) \cup \sce(s^2)) \subseteq \out$.
    \end{proof}

  \begin{proposition}
    Let $\phi = \langle X, Q, D, C\rangle$ be a QCSP in which value $a
    \in D_{x_i}$ is fixable for an existential variable $x_i$, and let $\phi'$ denote the same
    QCSP in which value $a$ is effectively fixed (\ie\ $\phi' =
    \langle X, Q, D', C\rangle$ where $D'_{x_i} = \{a\}$ and $D'_{x_j} =
    D_{x_j}, \forall j \not= i$). Then $\phi$ is true iff $\phi'$ is
    true.
  \end{proposition}
    
    \begin{proof} 
      If $\phi'$ has a winning strategy then the same strategy is also winning 
      for $\phi$; having $\phi'$ true therefore implies that $\phi$ is also true.
      
      On the other hand suppose that $\phi$ has
      a winning strategy $s^1$. That $a \in D_{x_i}$ is shallow-fixable for $x_i$
      means that we have:
      \[
      \forall t \in \out. ~\exists t' \in \out. ~
      (t|_{X_{i-1}} = t'|_{X_{i-1}} \wedge t'_{x_i} = a)
      \]
      
      The proof is similar to the one already detailed for Prop. 
      \ref{prop:correctness-of-removability}: 
      we show that if $s^1$ has a scenario $t \in \sce(s^1)$
      such that $t_{x_i} \not= a$, then we can ``correct'' this and exhibit another winning 
      strategy $s$ whose scenarios are the same as those of $s^1$ except 
      that all scenarios $\lambda$ such that $\lambda|_{X_{i-1}} = t|_{X_{i-1}}$
      have been replaced by tuples $t'$ with $t'_{x_i} = a$.
      More precisely, every scenario $t' \in \sce(s)$ will satisfy:

      \begin{itemize}
      
      \item If $t'|_{X_{i-1}} \not= t|_{X_{i-1}}$ then $t' \in \sce(s^1)$.
      
      \item If $t'|_{X_{i-1}} = t|_{X_{i-1}}$ then $t'_{x_i} = a$.
      
      \end{itemize}
      
      This will prove the result: in showing how to construct $s$ we show that,
      starting from any winning strategy $s^1$ containing a number $n > 0$ of
      ``incorrect'' scenarios $t'$ with $t'_{x_i} \not= a$, we can always exhibit a winning
      strategy with at most $n-1$ such scenarios. This shows that there exists a 
      winning strategy in which no tuple $t'$ is such that $t'_{x_i} \not= a$.
      
      Let us now see how to construct $s$ starting from $s^1$. The outcome $t \in \sce(s^1)$
      needs to be replaced. Using the shallow 
      fixability property, we know that there exists another outcome $\theta \in \out$ such that
      $\theta|_{X_{i-1}} = t|_{X_{i-1}} \wedge \theta_{x_i} = a$. This outcome belongs to
      at least one winning strategy. We choose one of these strategies, which we
      call $s^2$. To define the new strategy $s$ we must define the functions $s_{x_j}$, for
      each $x_j \in E$. These functions are defined as follows:
      
      \begin{itemize}
      
      \item if $j < i$ then $s_{x_j}$ is defined as $s^1_{x_j}$ (\eg\ we follow
      the strategy $s^1$ for the first variables, until variable $x_i$, excluded); 

      \item for the following variables, \ie\ when $j \geq i$, we define the value of 
      $s_{x_j}(\tau)$, for each $\tau \in \prod_{y \in A_{j-1}}D_y$, as follows:

        \begin{itemize}
        
        \item if $\tau|_{X_{i-1}} = t|_{X_{i-1}}$, then 
        $s_{x_j}(\tau) = s^2_{x_j}(\tau)$;
        
        \item if $\tau|_{X_{i-1}} \not= t|_{X_{i-1}}$, then 
        $s_{x_j}(\tau) = s^1_{x_j}(\tau)$;
        
        \end{itemize}
      
      \end{itemize}

      The proof is completed by checking that every scenario $t' \in \sce(s)$
      satisfies the two desired properties:

      \begin{itemize}
      
      \item If $t'|_{X_{i-1}} \not= t|_{X_{i-1}}$ then $t' \in \sce(s^1)$, because, for each
      $x_j \in E$, we have $t'_{x_j} = s_{x_j}(t'|_{A_{j-1}}) = s^1_{x_j}(t'|_{A_{j-1}})$ 
      in this case.
      
      \item If $t'|_{X_{i-1}} = t|_{X_{i-1}}$ then $t'_{x_i} = a$, because 
      $t'_{x_i} = s_{x_i}(t'|_{A_{i-1}}) = s^2_{x_i}(t'|_{A_{i-1}}) 
      = s^2_{x_i}(t|_{A_{i-1}})  = s^2_{x_i}(\theta|_{A_{i-1}}) = \theta_{x_i} = a$.
      
      \end{itemize}
      
      Furthermore, every $t' \in \sce(s)$ with $t'|_{X_{i-1}} = t|_{X_{i-1}}$
      belongs to $\sce(s^2)$, and $s$ is therefore a winning strategy:  
      $\sce(s) \subseteq (\sce(s^1) \cup \sce(s^2)) \subseteq \out$.
    \end{proof}

  \begin{proposition}
    Let $\phi = \langle X, Q, D, C\rangle$ be a QCSP in which value $a
    \in D_{x_i}$ is dual-shallow-removable for a universal variable $x_i$, 
    and let $\phi'$ denote the
    same QCSP in which value $a$ is effectively removed (\ie\ $\phi' =
    \langle X, Q, D', C\rangle$ where $D'_{x_i} = D_{x_i} \setminus
    \{a\}$ and $D'_{x_j} = D_{x_j}, \forall j \not= i$). Then $\phi$ is true
    iff $\phi'$ is true.
  \end{proposition}
  
  \begin{proof}
    Direct consequence of Prop. \ref{prop:correctness-of-removability}: 
    the hypothesis is that the dual-shallow-removability holds, \ie\ 
    $a$ is removable for $x_i$ \wrt\ the negated QCSP $\neg \phi$; 
    then $\phi$ is true iff $\neg \phi$ is false iff $\neg \phi'$ is
    false iff $\phi'$ is true.
  \end{proof}

  \begin{proposition}
    Let $\phi = \langle X, Q, D, C\rangle$ be a QCSP in which value $a
    \in D_{x_i}$ is dual-shallow-fixable for an universal variable $x_i$, 
    and let $\phi'$ denote the same
    QCSP in which value $a$ is effectively fixed (\ie\ $\phi' =
    \langle X, Q, D', C\rangle$ where $D'_{x_i} = \{a\}$ and $D'_{x_j} =
    D_{x_j}, \forall j \not= i$). Then $\phi$ is true iff $\phi'$ is
    true.
  \end{proposition}
  
  \begin{proof}
    Direct consequence of Prop. \ref{prop:correctness-of-fix}: 
    the hypothesis is that the dual-shallow-fixability holds, \ie\ 
    $a$ is fixable for $x_i$ \wrt\ the negated QCSP $\neg \phi$; 
    then $\phi$ is true iff $\neg \phi$ is false iff $\neg \phi'$ is
    false iff $\phi'$ is true.
  \end{proof}

  \begin{proposition}
    Let $\phi = \langle X, Q, D, C\rangle$ be a QCSP.
    Given a tuple $t \in \prod_{x \in X} D_x$,
    we denote by $B$ the conjunction of constraints:
    \begin{equation}
    \bigwedge_{x_i \in E}
    \left(
      \left(\bigwedge_{y \in A_{i-1}} y = t_y\right) 
      \rightarrow (x_i = t_{x_i})
    \right) 
    \label{eq-additional-constraints-for-outcome-bis}
    \end{equation}
    The QCSP $\psi = \langle X, Q, D, B \cup C\rangle$ is true
    iff $t \in \out^\phi$.
  \end{proposition}
    
    \begin{proof}
      Assume that $\psi$ is true. Then it has a non empty set of winning strategies; 
      let $s$ be one of them, picked arbitrarily. Let $t'$ be the scenario of $s$ 
      that is such that $t'|_A = t|_A$, \ie\ that assigns the same values
      as $t$ on the universal variables. Because $s$ is a winning strategy, $t'$ 
      is a solution, and it satisfies the constraint given 
      by \eqref{eq-additional-constraints-for-outcome-bis}. 
      A straightforward induction
      on the indices of the existential variables shows that $t$ is indeed 
      identical to $t'$, which implies $t \in \out^\phi$. 

      Assume now that $t \in \out^\phi$, \ie\ there exists a winning strategy $s$
      for $\phi$ such that $t \in \sce(s)$. Every scenario $t' \in \sce(s)$ satisfies
      $C$. Let us prove by case that each $t' \in \sce(s)$ also satisfies $B$. 
      If we consider the scenario $t'$ which is such that $t'|_A = t|_A$, then 
      this scenario is indeed $t$ (a strategy defines a unique outcome for each
      assignment of the universal variables), which satisfies $B$. 
      On the other hand, $B$ is satisfied also
      if we consider
      any tuple $t'$ which is such that $t'|_A \not= t|_A$. 
      To see this, 
      let $j$ be the lowest index such that $t'|_{A_{j-1}} \not= t|_{A_{j-1}}$.
      Constraints of $B$ with $i < j$ are satisfied because $t'|_{A_{i-1}} = t|_{A_{i-1}}$;
      the others because 
      %
      %then 
      the left-hand side
      of the implications $\left(\bigwedge_{y \in A_{i-1}} t'_y = t_y\right) 
      \rightarrow (t'_{x_i} = t_{x_i})$ are false. %, and $B$ is therefore verified.
      Every scenario of $s$ therefore satisfies $B \wedge C$, in other words this
      strategy is winning for $\psi$.
    \end{proof}

  \begin{proposition}
    Given a QCSP $\phi = \langle X, Q, D, C \rangle$, the problems of
    deciding whether:
    
    \begin{itemize}
    
    \item value $a \in D_{x_i}$ is \deepfixable, \deepremovable,
    \inconsistent, \implied\ for variable $x_i \in X$,
    
    \item value $a \in D_{x_i}$ is \deepsubstitutable\ to or
    \deepinterchangeable\ with $b \in D_{x_i}$ for variable $x_i \in X$,
    
    \item variable $x_i \in X$ is \dependent\ on variables $V \subseteq X$,
    or is \deepirrelevant,
    
    \end{itemize}

    \noindent
    are PSPACE-complete.
  \end{proposition}
    
    \begin{proof}(membership in PSPACE)
      The membership in PSPACE relies essentially on 
      Prop. \ref{prop-trick-for-pspace-upper-bound} and its immediate consequence,
      mentioned in the main text, that testing whether $t \in \out$ can be done in
      polynomial space. All properties hold iff some statement is verified for all 
      $t \in \out$, so the idea is then to loop over each tuple $t$, 
      determine whether it belongs to $\out$ and, if this is the case,
      check whether it satisfies the statement. 
      For inconsistency we check whether $t_{x_i} \not= a$. We return false
      as soon as we have a tuple $t \in \out$ for which this is not the case.
      For implication we test whether $t_{x_i} = a$ and similarly return false if
      one tuple does not verify that. The same idea works for all properties:
      for fixability we test whether $t[x_i := a] \in \out$;
      for substitutability we check whether       
      $(t_{x_i} = a) \implication (t[x_i:=b] \in \out)$;
      for removability we check whether 
      $(t_{x_i} = a)  \implication (\exists b \neq a.
          ~ t[x_i := b] \in \out)$;
      for determinacy we check whether 
      $\forall b \not= t_{x_i}. ~t[x_i := b] \not\in \out$;
      for irrelevance we check whether 
      $\forall b \in D_{x_i}. ~ t[x_i := b] \in \out$.
      For dependency we have to do a double loop in lexicographical order, check
      whether both tuples $t, t'$ belong to $\out$ and, if, so, check whether 
      $(\forall x_j \in V. ~ t_{x_j} = t'_{x_j})
      \implication (t_{x_i} = t'_{x_i})$. 
      In any case, at the end of the loop, we return true if no counter-example
      to the property has been found.
      It is clear that these algorithms use polynomial space and return true
      iff the considered property holds.
    \end{proof}
    
    \begin{proof}(hardness for PSPACE)
      For all properties we reduce the problem of deciding whether a 
      QCSP $\phi = \langle X, Q, D, C\rangle$ is false 
      to the problem of testing whether the considered property holds. 
      
      The reductions work as follows. For \textbf{inconsistency} we simply construct the
      QCSP $\psi = \langle X \cup \{x\}, Q', D', C\}\rangle$, where:
      \begin{itemize}
      
      \item $x$ is a fresh variable, \ie\ $x \not\in X$;
      
      \item $Q'$ is similar to $Q$ except that the new variable $x$ 
      is quantified existentially, \ie\
      $Q'_{y} = Q_{y}, ~ \forall y \not= x$ and $Q'_{x} = \exists$;
      
      \item $D'$ is similar to $D$ except that the domain of the new variable 
      $x$ is a singleton, \ie\
      $D'_{y} = D_{y}, ~ \forall y \not= x$ and $D'_{x} = \{a\}$ for some arbitrary $a$. 
      
      \end{itemize}
      
      It is straightforward that $\phi$ has a winning strategy iff $\psi$ also does. 
      Let us verify that $\phi$ is false iff value $a$ is inconsistent for
      variable $x$ in $\psi$: if $\phi$ is false then $\out^\phi$ is empty, and so
      is $\out^\psi$, and then it is true that
      $\forall t \in \out^\psi. ~ t_{x_i} \not= a$; if $a$ is inconsistent for 
      $x$ in $\psi$ then $\forall t \in \out^\psi. ~ t_{x_i} \not= a$, but no outcome can
      assign a value different from $a$ to variable $x_i$, hence $\out^\psi$ is empty
      and $\out^\phi$ is also empty.
      
      The same reduction works directly for \textbf{removability}: $\phi$ is false iff
      $a$ is removable from $x$ in $\psi$. 
      
      For \textbf{fixability}, \textbf{implication}, \textbf{substitutability},
      \textbf{interchangeability} and \textbf{irrelevance}, 
      the reduction is only slightly different; now we construct the QCSP:
      \[
      \psi = \langle X \cup \{x\}, Q', D', C \cup \{x = 0\}\}\rangle
      \] 
      in which the new variable $x$ is existential and ranges over $\{0, 1\}$. 
      Note that the constraint $x = 0$ can be expressed directly in each and every 
      of our 5 formalisms. We can check that $\phi$ is false iff:
      
      \begin{itemize}
      
      \item variable $x$ is fixable to value 1 in $\psi$: 
      if $\phi$ is false then $\out^\phi$ is empty and so is $\out^\psi$ and we trivially
      have $\forall t \in \out^\psi. ~t[x := 1] \in \out^\psi$; if 
      $x$ is fixable to 1 in $\psi$ then $\forall t \in \out^\psi. ~t[x := 1] \in \out^\psi$,
      but there is no $t$ is such that $t[x := 1] \in \out^\psi$ and $\out^\psi$ and
      $\out^\phi$ are empty. 
      
      \item value $1$ is implied for variable $x$ in $\psi$: similarly to
      fixability we have $\out^\phi = \emptyset$ iff 
      $\forall t \in \out^\psi. ~ t_{x} = 1$.
      
      \item value 0 is substitutable to value 1 for variable $x$ in $\psi$
      ($\out^\phi = \emptyset$ holds iff $\forall t \in \out^\psi. ~ 
      (t_{x} = 0) \implication (t[x:=1] \in \out^\psi)$).
      
      \item value 0 is interchangeable with value 1 for variable $x$ in $\psi$:
      ($\out^\phi = \emptyset$ holds iff $\forall t \in \out^\psi. ~ 
      (t_{x} = 0) \leftrightarrow (t[x:=1] \in \out^\psi)$).
      
      \item variable $x$ is irrelevant in $\psi$: if
      $\forall t \in \out^\psi. ~\forall b \in \{0,1\}. ~ t[x := b] \in \out^\psi$,
      then any $t \in \out^\psi$ is in particular such that $t[x := 1] \in \out^\psi$
      so no such $t$ exists and $\out^\psi = \emptyset$. (The other direction is
      trivial.)
      
      \end{itemize}
      
      For \textbf{determinacy} and \textbf{dependence}, 
      the reduction consists in constructing the QCSP
      $\psi = \langle X \cup \{x\}, Q', D', C\}\rangle$,
      in which the new variable $x$ is existential and ranges over $\{0, 1\}$. 
      
      We check that $\phi$ is false if $x$ is determined in $\psi$. Assume that
      $\forall t \in \out^\psi. ~\forall b \not= t_{x}. ~t[x := b] \not\in \out^\psi$,
      and let us consider an arbitrary $t \in \out^\psi$. Its value on $x$ is either 0 or
      1 (say 0). Then it is such that $t[x := 1] \not\in \out^\psi$. Because values 0 and 1 
      play a symmetric role, this cannot be, and $\out^\phi = \emptyset$.
      (The other implication is trivial.)
      
      We last check that $\phi$ is false if variable $x$ is dependent on the set of
      variables $X$ in $\psi$. Assume that
      $\forall t, t' \in \out. ~
         (t|_X = t'|_X) \implication (t_{x} = t'_{x})$.
      Let us consider an arbitrary tuple $t \in \out^\psi$ with (say) $t_x = 0$.
      If we consider the tuple $t' = t[x := 1]$, then this tuple is
      such that $t'|_X = t|_X$, and therefore does not belong to $\out^\psi$
      (if it did, then we'd have $t'_x = t_x$). 
      Because values 0 and 1 play a symmetric role, this cannot be, and  
      $\out^\phi = \emptyset$.
      (The other implication is trivial.)
      
      In all our reductions, we can start from any of the 5 formalisms listed in Sec. 
      \ref{subsec-encoding-issues}, and the resulting QCSP is expressed in the same formalism.
      It is well-known that deciding the truth of a QCSP in any of these formalisms is
      PSPACE-complete and the hardness result therefore holds in all 5 cases.
    \end{proof}

  \begin{proposition}
  
    Given a QCSP $\phi = \langle X, Q, D, C \rangle$, the problems of
    deciding whether:
    
    \begin{itemize}
    \item value $a \in D_{x_i}$ is \shallowfixable, \shallowremovable\ for
    variable $x_i \in X$,

    \item value $a \in D_{x_i}$ is \shallowsubstitutable\ to or
    \shallowinterchangeable\ with $b \in D_{x_i}$ for variable $x_i \in X$,

    \item variable $x_i \in X$ is \shallowirrelevant,
    
    \end{itemize}

    \noindent
    are PSPACE-complete.
  \end{proposition}
  
    \begin{proof}
      For membership in PSPACE the algorithm is similar to 
      Prop. \ref{prop:dproperties-PSPACEc}: we use the fact that
      testing whether $t \in \out$ can be done in polynomial space by
      Prop. \ref{prop-trick-for-pspace-upper-bound}. To check 
      whether a property of the form $\forall t \in \out. ~\gamma$ is true,
      we loop over all tuples in lexicographical order, test whether
      the current tuple is an outcome and, if so, verify that it satisfies $\gamma$.
      For properties of the form $\exists t \in \out. ~\gamma$, we do a similar loop 
      and return true iff one of the outcomes met during the loop satisfied $\gamma$. 
      This works in polynomial space for all properties.
      
      The hardness is a direct consequence of the fact that shallow properties
      are equivalent to the deep ones in the particular case when the variable on which the
      property is asserted is at the tail of the linearly ordered set of variables.
      In all the reductions used in the proof of Prop. \ref{prop:dproperties-PSPACEc},
      note that we introduce a variable that can be introduced \emph{at an arbitrary place}. 
      The reductions can therefore be directly adapted to the shallow definitions. 
      
      For instance, in the case of fixability, the reduction consisted, starting
      from a QCSP $\phi = \langle X, Q, D, C\rangle$, to construct the QCSP 
      $\psi = \langle X \cup \{x\}, Q', D', C \cup \{x=0\}\rangle$, with $D'_{x} = \{0,1\}$. 
      We consider the same reduction and impose that $x$ be placed at the end of 
      the ordered set $X$. Then $x$ is shallow-fixable to 1 iff it is deep-fixable to 1. 
      We have proved that $\phi$ is false if variable $x$ is deep-fixable to 1 in $\psi$,
      which is true if it is shallow-fixable to 1 in $\psi$. Similarly in all cases of 
      Prop. \ref{prop:dproperties-PSPACEc} the reduction directly
      applies to shallow property as long as we impose that the new variable $x$ be
      put at the end of the quantifier prefix.
    \end{proof}

  \begin{proposition}
    Given a $\Sigma_k$QCSP $\phi = \langle X, Q, D, C \rangle$ encoded using
    Formalism (I), the problems of deciding whether:
    
    \begin{itemize}
    
    \item value $a \in D_{x_i}$ is \textsl{deep-fixable},
    \textsl{deep-removable}, \inconsistent, \implied\ for
    variable $x_i \in X$,
    
    \item value $a \in D_{x_i}$ is \textsl{deep-substitutable} to or
    \textsl{deep-interchangeable} with $b \in D_{x_i}$ for variable
    $x_i \in X$,
    
    \item variable $x_i \in X$ is \dependent\ on variables $V \subseteq X$,
    or is \textsl{deep-irrelevant},
    
    \end{itemize}

    \noindent
    are $\Pi_k^p$-hard and belong to $\Pi_{k+1}^p$. Moreover, for deep inconsistency, 
    implication, determinacy and dependence, the problems are more precisely 
    $\Pi_k^p$-complete.
  \end{proposition}

    The use of formalism (1) means that $\phi$ is a Quantified Boolean Formula of the form:
    \[
    \phi :~~ 
    \exists M_1. ~ \forall M_2. \cdots Q_k M_k. ~ C
    \]
    where the $M_i$s are blocks of variables of alternating quantification, 
    $C$ is a Boolean circuit built on these variables,
    and the last block $M_k$ is quantified universally ($Q_k = \forall$) if 
    $k$ is even, and existentially ($Q_k = \exists$) if $k$ is odd.
    Consistently with previous notation, the 
    linearly ordered set $X = \{x_1 \dots x_n\}$ denotes the union of all variables 
    of the prefix, and the notations $E_j$, $A_j$, \etc, are defined as in Section 
    \ref{subsec:Definition-of-QCSP}. 
    
    For technical reasons it is more convenient to analyze the complexity of the
    \emph{negations} of these properties, \ie\ we focus on the complexity of determining
    whether the property \emph{does not hold}. So we prove that the negations
    are $\Sigma_k^p$-hard and belong to $\Sigma_{k+1}^p$. (The problem of testing whether
    a $\Sigma_k^p$QBF is false is $\Pi_k^p$-complete.)

    \begin{proof}(membership results)
      For \textbf{consistency}, membership in $\Sigma_{k}^{p}$ is shown as follows: we 
      are given a formula $\phi$ of the aforementioned form, as well as $a$ and $x_i$, and we
      want to test whether $\exists t \in \out^\phi. ~ t_{x_i} = a$. 
      We use a reduction similar to the one used by 
      Prop. \ref{prop-trick-for-pspace-upper-bound}, and construct a formula which
      is true iff the property holds.
      The formula used in Prop. \ref{prop-trick-for-pspace-upper-bound} imposes
      additional constraints whose role is to make sure that the outcome belongs to
      the set of scenarios of any winning strategy of the produced formula. In our case the
      outcome in question is quantified existentially and is of the form
      $\langle v_1, \dots, v_n \rangle$ with $v_i = a$. 
      We obtain the formula:
      \begin{equation}
        \psi : ~~ 
        \exists v_1, \dots v_n. ~
        \exists M_1. ~ \forall M_2 \cdots Q_k M_k.  ~ 
        (B \wedge C \wedge v_i = a)
        \label{eq-ph-main-formula}
      \end{equation}
      
      \noindent where each variable $v_i$ ranges over $D_{x_i}$ and $B$ is 
      the conjunction:
      \[
        \bigwedge_{x_i \in E}
        \left(
          \left(\bigwedge_{y_j \in A_{i-1}} y_j = v_j \right) 
          \rightarrow (x_i = v_i)
        \right) 
      \]
      
      Note that the existentially quantified variables
      $\langle v_1, \dots, v_n \rangle$ are not redundant with the $x_j$s: we want to 
      impose that \emph{at least} one of the outcomes of $\psi$ assign $x_i$ to $a$, whereas
      simply adding the constraint $x_i = a$ would enforce it for every scenario of 
      any strategy. Formula $\psi$ is true iff there exists a tuple 
      $t \in \out^\phi$ such that $t_{x_i} = a$ is a direct consequence of 
      Prop. \ref{prop-trick-for-pspace-upper-bound}. Formula $\psi$ is itself
      a $\Sigma_k$-QBF and we can therefore determine whether it is true in $\Sigma_k^p$.
      
      Non-\textbf{implication} ($\exists t \in \out. ~ t_{x_i} \not= a$),
      Eq. \ref{eq-ph-main-formula} is simply replaced by:
      \[
        \psi : ~~ 
        \exists v_1, \dots v_n. ~
        \exists M_1. ~ \forall M_2 \cdots Q_k M_k.  ~ 
        (B \wedge C \wedge \fbox{$v_i \not= a$})
      \]
      
      Non-\textbf{determinacy} is expressed as
      $\exists t \in \out. ~\exists b \not= t_{x_i}. ~t[x_i := b] \in \out$
      or, equivalently, as
      $\exists t \in \out. ~\exists t' \in \out. ~ 
      t'|_{X \setminus\{x_i\}} = t|_{X \setminus\{x_i\}}  
      \wedge t'_{x_i} \not= t_{x_i}$. 
      We have to assert the joint existence of the two outcomes $t$ and $t'$,
      whose values on variables $x_1 \dots x_n$ are noted 
      $\langle v_1, \dots, v_n \rangle$ and $\langle v'_1, \dots, v'_n \rangle$,
      respectively. We obtain:
      \[
        \psi : ~~ 
        \exists v_1, \dots v_n. ~
        \exists v'_1, \dots v'_n. ~
        \left(\begin{array}{rl}
          & \bigwedge_{j \not= i} v'_j = v_j 
          ~~ \wedge ~~
           v'_i \not= v_i
        \\ \wedge 
          & \exists M_1. ~ \forall M_2 \cdots  Q_k M_k. ~ 
            (B \wedge C )
        \\ \wedge 
          & \exists M'_1. ~ \forall M'_2 \cdots Q'_k M'_k. ~ 
            (B' \wedge C')
        \end{array}\right)
      \]
      Now we note that the two matrices $(B \wedge C)$ and $(B' \wedge C')$
      are imposed on disjoint sets of variables (the unprimed and the primes
      variables, respectively), so we can rewrite the 
      previous formula in a $\Sigma_k$ form, as follows:
      \[
        \psi : ~~ 
        \exists v_1, \dots v_n. ~
        \exists v'_1, \dots v'_n. ~
        \left(\begin{array}{l}
           \bigwedge_{j \not= i} v'_j = v_j 
          ~~ \wedge ~~
           v'_i \not= v_i
          ~~ \wedge 
        \\
           \exists M_1, M'_1. ~ \forall M_2, M'_2 
        \\  ~~~~~~~
           \cdots 
            Q_k M_k, M'_k. ~ 
            (B \wedge C \wedge B' \wedge C')
        \end{array}\right)
      \]
      
      Non-\textbf{dependence} can be stated as
      $\exists t \in \out. ~\exists t' \in \out. ~
           (t|_V = t'|_V) \wedge (t_{x_i} \not= t'_{x_i})$,
      relying on the fact that the domain only has two values;
      the proof is similar except that $\psi$ has the following form:
      \[
        \psi : ~~ 
        \exists v_1, \dots v_n. ~
        \exists v'_1, \dots v'_n. ~
        \left(\begin{array}{l}
          \fbox{$\bigwedge_{x_j \in V} v'_j = v_j$}
          ~~ \wedge ~~
           v'_i \not= v_i
          ~~ \wedge ~~
        \\
          \exists M_1, M'_1. ~ \forall M_2, M'_2 
        \\ ~~~~~~~
           \cdots 
            Q_k M_k, M'_k. ~ 
            (B \wedge C \wedge B' \wedge C')
        \end{array}\right)
      \]

      For the other properties it is less obvious to see whether the 
      upper bound of $\Sigma_k^p$ holds, because      
      their negations are defined as follows:

      \begin{itemize}
      
      \item Non-\textbf{fixability} can be expressed as 
      $\exists t \in \out. ~ t[x_i := a] \not\in \out$ or, equivalently, 
      $\exists t \in \out. ~ \exists t' \not\in \out. ~ 
      t'|_{X \setminus\{x_i\}} = t|_{X \setminus\{x_i\}}  
      \wedge t'_{x_i} = a$;
 
      \item Non-\textbf{substitutability} as 
      $\exists t \in \out. ~ (t_{x_i} = a) \wedge (t[x_i:=b] \not\in \out)$ or, equivalently, 
      $\exists t \in \out. ~ \exists t' \not\in \out. ~ 
      t'|_{X \setminus\{x_i\}} = t|_{X \setminus\{x_i\}}  
      \wedge t_{x_i} = a \wedge t'_{x_i} = b$;
      
      \item Non-\textbf{removability} as 
      $\exists t \in \out. ~(t_{x_i} = a) \wedge (\forall b \neq a. ~ t[x_i := b] \not\in \out)$
      or, equivalently, as
      $\exists t \in \out. ~ \forall t' \not\in \out. ~ 
      (t'|_{X \setminus\{x_i\}} = t|_{X \setminus\{x_i\}})  \rightarrow t_{x_i} = t'_{x_i}$;
%      (here again we use the fact that the domain the domain has only two values);

      \item Non-\textbf{irrelevance} is expressed as
      $\exists t \in \out. ~\exists t' \not\in \out. ~ 
      t'|_{X \setminus \{x_i\}} = t|_{X \setminus \{x_i\}}$.
      
      \end{itemize}
      
      The problem is that in each case we need to find both an outcome $t$ and another
      tuple $t'$ which is not an outcome. The quantifier pattern for asserting that
      $t$ is not an outcome is now of the form 
      $\forall M_1. ~ \exists M_2 \cdots \overline{Q_k} M_k$, where $\overline{Q_k}$ is 
      the dual quantifier to $Q_k$. For instance for non-irrelevance the obtained 
      formula has the following form:
      \[
        \exists v_1, \dots v_n. ~
        \exists v'_1, \dots v'_n. ~
        \left(\begin{array}{rl}
          & \bigwedge_{j \not= i} v'_j = v_j 
        \\ \wedge 
          & \exists M_1. ~ \forall M_2 \cdots  Q_k M_k. ~ 
            (B \wedge C )
        \\ \wedge 
          & \forall M'_1. ~ \exists M'_2 \cdots \overline{Q_k}_k M'_k. ~ 
            (\neg B' \vee \neg C')
        \end{array}\right)
      \]
      
      Similarly to before, the variables involved in the matrics
      $(B \wedge C )$ and $(\neg B' \vee \neg C')$ are disjoint and we can merge them
      into one prefix. We rename the indexing of the primed blocks as follows:
      \[
        \exists v_1, \dots v_n. ~
        \exists v'_1, \dots v'_n. ~
        \left(\begin{array}{rl}
          & \bigwedge_{j \not= i} v'_j = v_j 
        \\ \wedge 
          & \exists M_1. ~ \forall M_2 \cdots  Q_k M_k. ~ 
            (B \wedge C )
        \\ \wedge 
          & \forall M'_2. ~ \exists M'_3 \cdots Q_{k+1} M'_{k+1}. ~ 
            (\neg B' \vee \neg C')
        \end{array}\right)
      \]
      and obtain:
      \[
        \exists v_1, \dots v_n. ~
        \exists v'_1, \dots v'_n. ~
        \left(\begin{array}{l}
           \bigwedge_{j \not= i} v'_j = v_j 
           ~~ \wedge 
        \\
           \exists M_1. ~ \forall M_2, M'_2.  ~\exists M_3, M'_3. 
        \\ ~~~~~~~
            \cdots  Q_k M_k. ~ Q_{k+1} M_{k+1}~ 
            ((B \wedge C ) \wedge
            (\neg B' \vee \neg C'))
        \end{array}\right)
      \]
      which is in $\Sigma_{k+1}^p$ form.
      We obtain a similar formula with minor changes for  fixability,
      substitutability and removability.
    \end{proof}
      
    \begin{proof}(hardness)
      The hardness part is easy, because in all the reductions used in 
      Prop. \ref{prop:dproperties-PSPACEc}
      to show the PSPACE-hardness of the properties, we reduced the problem 
      of determining whether a QCSP $\phi$ is false to the problem of checking the 
      considered property for a new formula $\psi$. The new formula $\psi$ was 
      constructed by introducing a new existential variable and this variable could be
      added into any quantifier block. Because of that, we can always make sure
      that the quantifier prefix of $\psi$ follows the same alternation as the one of
      $\phi$, and we can therefore reduce the problem of determining whether a 
      $\Sigma_k^p$QBF is false to the problem of testing the considered property
      is verified by a $\Sigma_k^p$QBF. 
      
      For instance the reduction used to prove that inconsistency is PSPACE-complete
      was as follows: we reduced any QCSP 
      $\phi: ~ \exists M_1. ~ \forall M_2. ~ \cdots Q_k M_k. ~ C$ to the QCSP 
      $\psi: ~ \exists M'_1, \{x\}. ~ \forall M_2. ~ \cdots Q_k M_k. ~ C$ with $D_x = \{a\}$.
      We had not specified the precise existential block in which the new variable $x$
      was added because the proof was precisely independent of that. We can now impose that
      it be inserted in the first block $M_1$. This shows that we can reduce the problem of 
      falsity for $\Sigma_k^p$QBFs to the problem of inconsistency for $\Sigma_k^p$QBFs. 
      Similarly, all the other proofs can be directly adapted to bounded quantifier alternations.
    \end{proof}

  \begin{proposition}
    Let $\phi = \langle X, Q, D, C \rangle$ be a QCSP where $C =
    \{c_1, \dots, c_m\}$. We denote by $\phi_{k}$ the QCSP $\langle
    X, Q, D, \{c_k\} \rangle$ in which only the $k$-th constraint is
    considered. We have, for all $x_i \in X$, $V \subseteq X$, 
    and $a, b \in D_{x_i}$:

    \begin{itemize}
    
    \item $\betweenparenth{\bigvee_{k \in 1 .. m} \inconsistent^{\phi_k}(x_i, a)}
      \implication  \inconsistent^{\phi}(x_i, a)$;

    \item $\betweenparenth{\bigvee_{k \in 1 .. m} \implied^{\phi_k}(x_i, a)}
      \implication  \implied^{\phi}(x_i, a)$;
      
    \item $\betweenparenth{\bigwedge_{k \in 1 .. m} \deepfixable^{\phi_k}(x_i, a)}
      \implication  \deepfixable^{\phi}(x_i, a)$;
      
    \item $\betweenparenth{\bigwedge_{k \in 1 .. m} \deepsubstitutable^{\phi_k}(x_i, a, b)}
      \implication  \deepsubstitutable^{\phi}(x_i, a, b)$;
      
    \item $\betweenparenth{\bigwedge_{k \in 1 .. m} \deepinterchangeable^{\phi_k}(x_i, a, b)}
      \implication  \deepinterchangeable^{\phi}(x_i, a, b)$;
      
    \item $\betweenparenth{\bigvee_{k \in 1 .. m} \determined^{\phi_k}(x_i)}
      \implication  \determined^{\phi}(x_i)$;
      
    \item $\betweenparenth{\bigwedge_{k \in 1 .. m} \deepirrelevant^{\phi_k}(x_i)}
      \implication  \deepirrelevant^{\phi}(x_i)$;
      
    \item $\betweenparenth{\bigvee_{k \in 1 .. m} \dependent^{\phi_k}(V, x_i)}
      \implication  \dependent^{\phi}(V, x_i)$.
      
    \end{itemize} 
  \end{proposition}

    \begin{proof}
      These propositions rely on the following
      \emph{monotonicity} property of the set of outcomes: if we have
      two QCSPs $\phi_1 = \langle X, Q, D, C_1\rangle$ and $\phi_2 =
      \langle X, Q, D, C_2\rangle$ (with the same quantifier prefix)
      and if $\sol^{\phi_1} \subseteq \sol^{\phi_2}$ 
      then 
      $\out^{\phi_1} \subseteq \out^{\phi_2}$. This is easy to see:
      any winning strategy $s$ for $\phi_1$ is such that 
      $\sce(s) \subseteq \sol^{\phi_1}$. Then it is also such that 
      $\sce(s) \subseteq \sol^{\phi_2}$ and it is a winning strategy for 
      $\phi_2$. 
      
      The proofs for inconsistency, implication and determinacy 
      directly follow:
      
      \begin{itemize}
      
      \item For inconsistency: if for some $k$ we have $\forall t \in
      \out^{\phi_k}. ~ t_{x_i} \not= a$, then we also have $\forall
      t \in \out^{\phi}. ~ t_{x_i} \not= a$,
      because $\out^{\phi} \subseteq \out^{\phi_k}$. 

      \item
      For implication: if for some $k$ we have $\forall t \in
      \out^{\phi_k}. ~ t_{x_i} = a$, then we also have $\forall
      t \in \out^{\phi}. ~ t_{x_i} = a$,
      because $\out^{\phi} \subseteq \out^{\phi_k}$. 
      
      \item
%       For determinacy: if for some $k$ we have $\forall t \in
%       \out^{\phi_k}. ~ t_{x_i} \not\in \out^{\phi_k}$, then we also have $\forall
%       t \in \out^{\phi}. ~ t_{x_i} \not\in \out^{\phi_k} \supseteq \out^{\phi}$. 
%      
      For determinacy: if for some $k$ we have $\forall t \in
      \out^{\phi_k}. ~ \forall b \not= t_{x_i}. ~ t[x_i := b] \not\in \out^{\phi_k}$, 
      then we also have $\forall
      t \in \out^{\phi}. ~ \forall b \not= t_{x_i}. ~ t[x_i := b] \not\in \out^{\phi} \subseteq \out^{\phi_k}$.

      \item
      For dependence: if for some $k$ we have $\forall t, t' \in \out^{\phi_k}. ~
      t|_V = t'|_V \implication\ t_{x_i} = t'_{x_i}$, then we also have
      $\forall t, t' \in \out^{\phi}. ~
      t|_V = t'|_V \implication\ t_{x_i} = t'_{x_i}$
      because $\out^{\phi} \subseteq \out^{\phi_k}$. 

      \end{itemize}

      Consider now deep fixability. We assume that forall $k$ and forall
      $t \in \out^{\phi_k}$ we have $t[x_i:=a] \in \out^{\phi_k}$. We consider
      a tuple $t \in \out^{\phi}$; since $\out^{\phi} \subseteq \out^{\phi_k}$ for all $k$,
      $t$ belongs to every $\out^{\phi_k}$, 
      and therefore $t[x_i := a]$ belongs to every $\out^{\phi_k}$ and therefore to every
      $\sol^{\phi_k}$. We conclude that $t[x_i := a] \in \sol^{\phi} = \bigcap_k \sol^{\phi_k}$. 
      We have seen in Prop. \ref{prop-no-outcome} that deep fixability
      can be stated as $\forall t \in \out^\phi. ~ t[x_i := a] \in \sol^\phi$, 
      which completes the proof.    

      For deep substitutability. We assume that forall $k$ and forall
      $t \in \out^{\phi_k}$ we have $t_{x_i} = a \implication\ t[x_i := b] \in \out^{\phi_k}$.
      We consider a tuple $t \in \out^{\phi}$ such that $t_{x_i}=a$; 
      since $\out^{\phi} \subseteq \out^{\phi_k}$ for all $k$,
      $t$ belongs to every $\out^{\phi_k}$, 
      and therefore $t[x_i := b]$ belongs to every $\out^{\phi_k}$ and therefore to every
      $\sol^{\phi_k}$. We conclude that $t[x_i := b] \in \sol^{\phi} = \bigcap_k \sol^{\phi_k}$. 
      We have seen in Prop. \ref{prop-no-outcome} that deep substitutability
      can be stated as $\forall t \in \out^\phi. ~ t_{x_i} = a \implication t[x_i := b] \in \sol^\phi$, 
      which completes the proof. 

      For deep interchangeability the result follows since two values $a$ and $b$
      are interchangeable iff $a$ is substitutable to $b$ and $b$ is substitutable to $a$.

      For deep irrelevance we use a result of Prop. \ref{prop-classifications}: 
      variable $x_i$ is irrelevant iff it is fixable to any value $a \in D_{x_i}$. 
      If forall $k$ we have $\deepirrelevant^{\phi_k}(x_i)$ then we have,
      forall $k$ and forall $a \in D_{x_i}$, $\deepfixable^{\phi_k}(x_i, a)$. It follows that,
      forall $a \in D_{x_i}$, $\deepfixable^{\phi}(x_i, a)$.
      This is equivalent to $\deepirrelevant^{\phi}(x_i)$.
    \end{proof}

\end{document}